\newif\ifarxiv
\arxivtrue

\ifarxiv
\documentclass[11pt]{article}
\usepackage{graphicx}
\usepackage[hyphens]{url}
\usepackage[english]{babel}
\usepackage{times}
\usepackage[margin=1in]{geometry}
\usepackage{enumitem}
\usepackage{hyperref}
\setlist{nosep}
\newcommand{\citet}{\cite}
\newcommand{\citep}{\cite}
\else
\documentclass[letterpaper]{article} 
\usepackage{aaai23}  
\usepackage{times}  
\usepackage{helvet}  
\usepackage{courier}  
\usepackage[hyphens]{url}  
\usepackage{graphicx} 
\usepackage{natbib}  
\usepackage{caption} 
\fi

\usepackage{algorithm}
\usepackage{algorithmic}

\ifarxiv
\else
\usepackage{newfloat}
\usepackage{listings}
\DeclareCaptionStyle{ruled}{labelfont=normalfont,labelsep=colon,strut=off} 
\floatstyle{ruled}
\newfloat{listing}{tb}{lst}{}
\floatname{listing}{Listing}
\fi

\usepackage{amsmath,amsthm,amssymb}
\usepackage{cleveref}
\usepackage{xspace}
\usepackage{enumitem}
\usepackage[normalem]{ulem}
\setlist{nosep}

\ifarxiv
\hypersetup{
    colorlinks,
    linkcolor={red},
    citecolor={blue},
    urlcolor={teal}
}
\fi

\newcommand{\nc}{\newcommand}
\nc{\DMO}{\DeclareMathOperator}
\nc\todo[1]{\textcolor{red}{[TODO: #1]}}
\nc{\st}{\star}
\nc\m[2]{m_{#1}(#2)}
\nc{\ReduceTree}{\texttt{ReduceTree}\xspace}
\nc{\PolyPriLearn}{\texttt{PolyPriLearn}\xspace}
\nc{\PPPLearn}{\texttt{PolyPriPropLearn}\xspace}
\nc{\GenericLearner}{\texttt{GenericLearner}\xspace}
\nc{\BR}{\mathbb{R}}
\nc{\BM}{\mathbb{M}}
\nc{\BT}{\mathbb{T}}
\nc{\BN}{\mathbb{N}}
\nc{\BZ}{\mathbb{Z}}
\nc{\Z}{\mathbb{Z}}
\nc{\ep}{\varepsilon}
\DMO{\height}{ht}
\renewcommand{\epsilon}{\varepsilon}
\nc{\ra}{\rightarrow}
\DMO{\Err}{err}
\DMO{\Opt}{opt}
\DMO{\Est}{Est}
\DMO{\good}{good}
\DMO{\negpt}{neg-pt}
\DMO{\VV}{{V}}
\DMO{\LL}{{L}}
\DMO{\finsupp}{fin}
\DMO{\supp}{supp}
\nc{\fin}{{\finsupp}}
\nc{\err}[2]{\Err_{#1}(#2)}
\nc{\rca}{\mathscr{B}}
\nc{\bt}{b}
\nc{\hMLp}{\hat\ML'}

\nc{\Rprot}{R}
\nc{\Sprot}{S}
\nc{\Aprot}{A}
\nc{\Pprot}{P}
\nc{\Pdist}{D}
\nc{\Qdist}{F}
\nc{\pdist}{d}
\nc{\qdist}{f}

\nc{\DP}{differentially private\xspace}

\nc{\SD}{\mathscr{D}}
\nc{\la}{\lambda}
\DMO{\KL}{KL}
\DMO{\Unif}{Unif}
\nc{\nn}{\varnothing}
\DMO{\SOA}{SOA}
\nc{\soa}[2]{\SOA_{#1}(#2)}
\nc{\soaf}[1]{\SOA_{#1}}
\nc{\gRes}[2]{\hat\MG({#1},{#2})}
\nc{\emp}{\hat P_{S_n}}
\DMO{\Red}{red}
\DMO{\Irred}{irred}
\nc{\Ired}{I^{\Red}}
\nc{\Iirred}{I^{\Irred}}

\DMO{\ssmp}{ssmp}
\DMO{\agg}{agg}
\DMO{\final}{final}
\DMO{\children}{children}

\nc{\pp}{p}
\nc{\PP}{P}
\nc{\QQ}{Q}
\nc{\DD}{D}

\DMO{\RAPPOR}{{RAPPOR}}
\nc{\RAP}{\RAPPOR}
\DMO{\RR}{RR}
\nc{\MD}{\mathcal{D}}
\nc{\ML}{\mathcal{L}}
\nc{\di}{P}
\nc{\MO}{\mathcal{O}}
\nc{\MM}{\mathcal{M}}
\nc{\MZ}{\mathcal{Z}}
\nc{\MU}{\mathcal{U}}
\nc{\MP}{\mathcal{P}}
\nc{\poly}{\mathrm{poly}}
\DMO{\treesum}{TreeSum}
\DMO{\lapsum}{LapSum}
\DMO{\checksum}{CheckSum}
\nc{\MDts}{\MD_{\treesum}}
\nc{\MDls}{\MD_{\lapsum}}
\nc{\MDcs}{\MD_{\checksum}}
\nc{\MC}{\mathcal{C}}
\nc{\MT}{\mathcal{T}}
\nc{\MS}{\mathcal{S}}
\nc{\MX}{\mathcal{X}}
\nc{\MY}{\mathcal{Y}}
\nc{\MA}{\mathcal{A}}
\nc{\MB}{\mathcal{B}}
\nc{\MJ}{\mathcal{J}}
\nc{\MF}{\mathcal{F}}
\nc{\MG}{\mathcal{G}}
\nc{\MQ}{\mathcal{Q}}
\nc{\p}{\Pr}
\nc{\E}{\mathbb{E}}
\nc{\tablesize}{s}
\DMO{\Hist}{hist}
\DMO{\Reg}{Reg}
\DMO{\prdim}{PRDim}
\nc{\eps}{\epsilon}
\nc{\hist}{\mathrm{hist}}

\nc{\ba}{\mathbf{a}}
\nc{\bx}{\mathbf{x}}
\nc{\bs}{\mathbf{s}}
\nc{\bv}{\mathbf{v}}
\nc{\bu}{\mathbf{u}}
\nc{\bw}{\mathbf{w}}
\nc{\by}{\mathbf{y}}
\nc{\bz}{\mathbf{z}}
\nc{\bp}{\mathbf{p}}
\nc{\bq}{\mathbf{q}}
\nc{\bn}{\mathbf{n}}
\nc{\bX}{\mathbf{X}}
\nc{\ind}{\mathbf{1}}
\nc{\hba}{\hat{\ba}}
\nc{\hbx}{\hat{\bx}}
\nc{\tbs}{\tilde{\bs}}
\nc{\hbs}{\hat{\bs}}
\nc{\hby}{\hat{\by}}
\nc{\bP}{\mathbf{P}}

\DMO{\sr}{sr}
\DMO{\Med}{Med}
\DMO{\Ber}{Ber}
\DMO{\Bin}{Bin}
\DMO{\Had}{Had}

\nc{\ME}{\mathcal{E}}
\DMO{\View}{View}
\nc{\B}{B}
\nc{\M}{M}
\nc{\ha}{\kappa}
\nc{\hk}{k}
\DMO{\pre}{pre}
\nc{\MH}{\mathcal{H}}
\DMO{\Ldim}{Ldim}
\DMO{\SQdim}{SQdim}
\DMO{\Tdim}{Tdim}
\DMO{\sfat}{sfat}
\DMO{\fat}{fat}
\DMO{\vc}{VCdim}
\DMO{\FO}{FO}
\DMO{\CM}{CM}
\nc{\MW}{\mathcal{W}}
\nc{\MV}{\mathcal{V}}
\nc{\MK}{\mathcal{K}}
\nc{\MN}{\mathcal{N}}

\nc{\BB}{\{0,1\}}
\nc{\bW}{\mathbf{W}}
\nc{\eell}{\ell}
\nc{\EELL}{L}
\nc{\q}{q}
\DMO{\size}{size}
\DMO{\emd}{EMD}
\DMO{\remd}{REMD}
\DMO{\tv}{TV}
\nc{\ts}{\tilde{s}}
\nc{\tc}{\tilde{c}}
\nc{\tH}{\tilde{H}}
\nc{\tG}{\tilde{G}}
\nc{\tgamma}{\tilde{\gamma}}
\nc{\bone}{\mathbf{1}}
\nc{\N}{\mathbb{N}}
\DMO{\scr}{scr}
\nc{\cP}{\mathcal{P}}
\nc{\cH}{\mathcal{H}}
\nc{\cA}{\mathcal{A}}
\nc{\cM}{\mathcal{M}}
\nc{\cT}{\mathcal{T}}
\nc{\bA}{\mathbf{A}}
\nc{\bB}{\mathbf{B}}
\nc{\bC}{\mathbf{C}}
\nc{\br}{\mathbf{r}}
\nc{\bh}{\mathbf{h}}
\nc{\hbh}{\hat{\bh}}
\nc{\R}{\mathbb{R}}
\nc{\A}{\mathbb{A}}
\nc{\bzero}{\mathbf{0}}
\DMO{\Lap}{Lap}
\DMO{\Polya}{Polya}
\nc{\bnu}{\mathbf{\nu}}
\nc{\cI}{\mathcal{I}}
\nc{\cJ}{\mathcal{J}}

\DeclareMathOperator*{\argmin}{arg\,min}
\DeclareMathOperator{\cost}{cost}
\DeclareMathOperator{\TV}{TV}
\DeclareMathOperator{\SIM}{SIM}
\DeclareMathOperator{\vol}{Vol}

\newcommand{\unsim}{\mathord{\sim}}
\newcommand{\salicon}{\textsc{salicon}\xspace}

\makeatletter
\newtheorem*{rep@theorem}{\rep@title}
\newcommand{\newreptheorem}[2]{%
\newenvironment{rep#1}[1]{%
 \def\rep@title{#2~\ref{##1}}%
 \begin{rep@theorem}}%
 {\end{rep@theorem}}}
\makeatother

\newtheorem{theorem}{Theorem}[section]
\newtheorem{corollary}[theorem]{Corollary}

\newtheorem{lemma}[theorem]{Lemma}
\newtheorem{informal theorem}[theorem]{Informal Theorem}

\newreptheorem{lemma}{Lemma}

\theoremstyle{definition}
\newtheorem{defn}{Definition}[section]

\usepackage{xcolor}
\usepackage{cancel}

\newcount\Comments
\newcommand{\badih}[1]{\ifnum\Comments=1\textcolor{red}{[Badih: #1]}\fi}
\newcommand{\pasin}[1]{\ifnum\Comments=1\textcolor{red}{[Pasin: #1]}\fi}
\newcommand{\ravi}[1]{\ifnum\Comments=1\textcolor{cyan}{[Ravi: #1]}\fi}

\allowdisplaybreaks

\title{Differentially Private Heatmaps}
\ifarxiv
\author{
 Badih Ghazi \hspace*{0.5cm}
 Junfeng He \hspace*{0.5cm}
 Kai Kohlhoff \hspace*{0.5cm}
 Ravi Kumar
 \\
 Pasin Manurangsi \hspace*{0.5cm}
 Vidhya Navalpakkam \hspace*{0.5cm}
 Nachiappan Valliappan
 \vspace{0.2cm}
 \\
 Google Research \\
 \texttt{\{badihghazi, ravi.k53\}@gmail.com} 
 \\
 \texttt{\{junfenghe,kohlhoff,pasin,vidhyan,nac\}@google.com}
}
\date{}
\else
\author{
  Badih Ghazi, Junfeng He, Kai Kohlhoff, Ravi Kumar\\Pasin Manurangsi, Vidhya Navalpakkam, Nachiappan Valliappan
}
\affiliations {
  Google Research, Mountain View, CA\\
  \texttt{\{badihghazi, ravi.k53\}@gmail.com} \\ \texttt{\{junfenghe,kohlhoff,pasin,vidhyan,nac\}@google.com}
}
\fi

\begin{document}

\maketitle

\begin{abstract}
We consider the task of producing heatmaps from users' aggregated data while protecting their privacy. We give a differentially private (DP) algorithm for this task and demonstrate its advantages over previous algorithms on real-world datasets.

Our core algorithmic primitive is a DP procedure that takes in a set of distributions and produces an output that is close in Earth Mover's Distance to the average of the inputs. We prove theoretical bounds on the error of our algorithm under a certain sparsity assumption and that these are near-optimal.  

\end{abstract}

\section{Introduction}
\label{sec:intro}

Recently, differential privacy (DP) \cite{dwork2006calibrating,dwork2006our} has emerged as a strong notion of user privacy for data aggregation and machine learning, with practical deployments including the 2022 US Census \cite{abowd2018us}, in industry \cite{erlingsson2014rappor,CNET2014Google, greenberg2016apple,dp2017learning, ding2017collecting} and in popular machine learning libraries \cite{tf-privacy, pytorch-privacy}.  Over the last few years, DP algorithms have been developed for several analytic tasks involving aggregation of user data.  

One of the basic data aggregation tools is a heatmap.  Heatmaps are popular for visualizing aggregated data in two or higher dimensions. They are widely used in many fields including computer vision and image processing, spatial data analysis, bioinformatics, etc.  Many of these applications involve protecting the privacy of user data.  For example, heatmaps for gaze or gene microdata~\cite{liu2019differential, steil2019privacy} would be based on data from individuals that would be considered private.  Similarly, a heatmap of popular locations in a geographic area will be based on user location check-ins, which are sensitive. 
Motivated by such applications, in this paper, we present an efficient DP algorithm for computing heatmaps with provable guarantees, and evaluate it empirically.


At the core of our algorithm is a primitive solving the following basic task: how to privately aggregate sparse input vectors with a small error as measured by the Earth Mover's Distance (EMD)? While closely related to heatmaps, the EMD measure is of independent interest: it was originally proposed for computer vision tasks~\cite{rubner2000earth} since it matches perceptual similarity better than other measures such as $\ell_1$, $\ell_2$, or KL-divergence~\cite{Stricker, levina2001earth, wang2012supervised}.  It is also well-suited for spatial data analysis since it takes the underlying metric space into account and considers ``neighboring'' bins.   EMD is used in spatial analysis~\cite{spatialEMD}, human mobility~\cite{MobilityEMD}, image retrieval~\cite{RubnerEtAl98, PuzichaEtAl99}, face recognition~\cite{XuEtAl08}, visual tracking~\cite{ZhaoEtAl04}, shape matching~\cite{GraumanDarrell04}, etc. For the task of sparse aggregation under EMD, we give an efficient algorithm with asymptotically tight error. We next describe our results in more detail.

\subsection{Our Results}

We consider the setting where each user $i$ holds a probability distribution $\bp_i$ over points in $[0, 1)^2$, and the goal is to compute the heatmap of the average of these probabilities, i.e., $\frac{1}{n} \sum_{i=1}^n \bp_i$. We give an $\eps$-DP algorithm for this task, establish its theoretical guarantees, and provide empirical evaluations of its performance.  (For definitions, see \Cref{sec:notation}.) 

\paragraph{Sparse Aggregation under EMD.}
At the heart of our approach is the study of aggregation under EMD\footnote{For a formal definition of EMD, please see \Cref{sec:prelim-emd}.}, where we would like to output the estimate of $\frac{1}{n} \sum_{i=1}^n \bp_i$ with the error measured in EMD. There are two main reasons why we consider EMD for the error measure. First, a bound on the EMD to the average distribution implies bounds on several metrics commonly used in evaluating heatmaps, including the KL-divergence, 
$\ell_1$ distance, and EMD itself. 
\ifarxiv
(We provide more details on this in Appendix~\ref{app:rel}.)
\fi
Second, while it is possible to obtain DP aggregation algorithms with bounded EMD error, as we will discuss below, any DP aggregation algorithm must suffer errors under other metrics, including KL-divergence or $\ell_1$ distance, that grow with the resolution\footnote{Specifically, it follows from previous work~\cite{DworkSSUV15} that, if we consider the $\ell_1$ distance or KL-divergence for $\Delta \times \Delta$ grid and $n \leq O_{\eps}(\Delta)$, then the error must be $\Omega(1)$.}, rendering them impractical when the number of users is small compared to the resolution.

When the distributions $\bp_i$'s are arbitrary, we show that a simple $\eps$-DP algorithm yields a guarantee of $O_\eps(1/\sqrt{n})$ on EMD, and that this bound is essentially optimal. While this is already a reasonable bound, we can improve on it by exploiting a property that is commonly present in  distributions used for aggregations: ``sparsity''~\cite{CormodeEtAl}.

Following the literature on compressed sensing~\cite{indyk2011k,BackursIRW16}, we define our approximation guarantee for the \emph{sparse EMD aggregation} problem with respect to the best $k$-sparse distribution\footnote{A distribution is \emph{$k$-sparse} if it is non-zero on at most $k$ points.} that approximates the average $\ba := \frac{1}{n} \sum_{i=1}^n \bp_i$ under EMD. More formally, we say that an output distribution $\hba$ is a \emph{$(\lambda, \kappa)$-approximation} for sparse EMD aggregation if 
\begin{align*}
\textstyle
\emd\left(\hba, \ba\right) \leq \lambda \cdot \min_{k\text{-sparse } \ba'} \emd\left(\ba', \ba\right) + \kappa,
\end{align*}
where $\lambda, \kappa > 0$ denote the multiplicative approximation ratio and additive error respectively.

Our main algorithmic contribution is in showing that under such a sparse approximation notion, we can achieve an error of only $O_{\eps}(\sqrt{k} / n)$ and that this is tight.\footnote{Note that the output $\hba$ need \emph{not} be $k$-sparse. This is the reason why the approximation ratio $\lambda$ can be less than one.}

\begin{theorem}[Informal] \label{thm:main-informal}
There exists an $\eps$-DP algorithm that, for any constant $\lambda \in (0, 1)$, can output a $(\lambda, O_{\eps}(\sqrt{k} / n))$-approximation for sparse EMD aggregation w.p. $0.99$. Furthermore, no $\eps$-DP algorithm can output a $(\lambda, o_{\eps}(\sqrt{k} / n))$-approximate solution w.p. $0.1$.
\end{theorem}

Due to a known connection between sparse EMD aggregation and $k$-median clustering on the plane~\cite{indyk2011k,BackursIRW16}, our result also yields an improved DP algorithm for the latter. 
\ifarxiv
We defer the formal statement of our $k$-median results and discussion to Appendix~\ref{app:kmedian}. 
\else
Due to space constraints, we omit the formal statement of our $k$-median results.
\fi

\paragraph{Experimental results.} We test our algorithm on both real-world location datasets and synthetic datasets. The results demonstrate its practicality even for moderate values of $\eps \in [0.5, 5]$ and a number of users equal to $200$. Furthermore, we compare our algorithm with simple baselines; under popular metrics for heatmaps, our results demonstrate significant improvements on these regimes of parameters.

\subsection{Overview of Techniques}

At a high level, our algorithm is largely inspired by the work of~\citet{indyk2011k} on \emph{compressed sensing} under EMD. Roughly speaking, in compressed sensing, there is an underlying vector $\bx$ that is known to be well-approximated by a sparse vector; we have to provide a matrix $\bA$ such that, when we observe the measurements $\bA\bx$, we can reconstruct $\bx'$ that is close to $\bx$ (under a certain metric). This can of course be done trivially by taking $\bA$ to, e.g., be the identity matrix. Thus, the objective is to perform this recovery task using as few measurements (i.e., number of rows of $\bA$) as possible.  There is a rich literature on compressive sensing; most relevant to our work are the prior papers studying compressive sensing with EMD, in particular,~\citet{indyk2011k} and~\citet{BackursIRW16}.


\citet{indyk2011k} presented an elegant framework for reducing the compressed sensing problem under EMD to one under $\ell_1$, which is well-studied\ifarxiv~(see, e.g.~\cite{berinde2008combining,berinde2008practical,IndykR08,berinde2009sequential})\else~\cite[see, e.g.,][]{berinde2008combining,berinde2008practical,IndykR08,berinde2009sequential}\fi. Their reduction centers around finding a linear transformation with certain properties. Once such a transformation is specified, the algorithm proceeds (roughly) as follows: transform the input $\bx$, run the compressed sensing scheme for $\ell_1$, and ``invert'' the transformation to get $\bx'$. Note that the number of measurements required is that of the $\ell_1$ compressed sensing scheme.

One can try to use the Indyk--Price scheme for DP aggregation by viewing the hidden vector $\bx$ as the sum $\sum_{i=1}^n \bp_i$, and then adding Laplace noise to each measurement to ensure privacy. Although they did not analyze their guarantees for noisy measurements, one can follow the robustness of known $\ell_1$ compressed sensing schemes to analyze the error.  Unfortunately, since the error will scale according to the $\ell_1$ norm of the noise vector and the noise vector consists of $O(k \cdot \log (n/k))$ entries, this approach only provides an error guarantee of $O(k \cdot (\poly\log n) / n)$.

To overcome this, we observe that, while compressed sensing and DP aggregation seem similar, they have different goals: the former aims to minimize the \emph{number} of measurements whereas the latter aims to minimize the \emph{error} due to the noise added (irrespective of the number of measurements). With this in mind, we proceed by using the Indyk--Price framework but \emph{without compressing}, i.e., we simply measure the entire transformation. Even with this, the noise added to achieve DP is still too large and makes the error dependent on $\log n$. As a final step, to get rid of this factor we carefully select a different noise magnitude for each measurement, which allows us to finally achieve the $O(\sqrt{k})$ error as desired.  The details  are presented in \Cref{sec:algo-main}.


Our lower bound follows the packing framework of~\citet{HardtT10}. Specifically, we construct a set of $k$-sparse distributions whose pairwise EMDs are at least $\Omega(1/\sqrt{k})$. The construction is based on an $\ell_1$ packing of the $\sqrt{k} \times \sqrt{k}$ grid, which gives a set of size $2^{\Omega(k)}$. It then immediately follows from~\citet{HardtT10} that the error must be at least $\Omega_{\eps}(\sqrt{k} / n)$ with probability $0.9$.
\ifarxiv
For more details, see Appendix~\ref{app:sparselb}
\fi


%

\subsection{Related Work and Discussion}

In a concurrent and independent work,~\citet{brella-heatmaps} also study the  private heatmaps problem.  However, our work differs from theirs in three aspects: (i) they do not formulate the problem in terms of EMD, (ii) their work does not provide any formal utility guarantees unlike ours, (iii) their emphasis is on communication efficiency in distributed/federated setting whereas our focus is more general.

Our DP EMD sparse aggregation algorithm bears high-level similarity to known algorithms for DP hierarchical histograms\ifarxiv~(see e.g.~\cite{CormodePSSY12,QardajiYL13})\else~\cite[see, e.g.,][]{CormodePSSY12,QardajiYL13}\fi: all algorithms may be viewed as traversing the grid in a top-down manner, starting with larger subgrids and moving on to smaller ones, where a noise is added to the ``measurement'' corresponding to each subgrid. The differences between the algorithms are in the amount of noise added to each step and how the noisy measurement is used to reconstruct the final output. Our choices of the noise amount and the Indyk–Price reconstruction algorithm are crucial to achieve the optimal EMD error bound stated in \Cref{thm:main-informal}.

There are also DP hierarchical histogram algorithms that do not fit into the above outline, such as the PrivTree algorithm~\cite{ZhangXX16}. An advantage of our approach is that the only aggregation primitive required is the Laplace mechanism; therefore, while we focus on the \emph{central} model of DP (where the analyzer can see the raw input and only the output is required to be DP), our algorithm extends naturally to distributed models that can implement the Laplace mechanism, including the secure aggregation model and the shuffle model~\cite{balle_merged,ghazi2019private}. On the other hand, algorithms such as PrivTree that use more complicated primitives cannot be easily implemented in these models.



\section{Notation and Preliminaries}
\label{sec:notation}

For $N \in \N \cup \{0\}$, we write $[N]$ to denote $\{0, \dots, N\}$. Let $G_\Delta$ be the set of $(\Delta \times \Delta)$ grid points in $[0, 1)^2$; specifically, $G_\Delta = \{(i/\Delta, j/\Delta) \mid i, j \in [\Delta - 1]\}$. For notational convenience, we assume throughout that $\Delta = 2^\ell$ for some $\ell \in \N$.

For an index set $\cI$, we view $\bp \in \R^\cI$ as a vector indexed by $\cI$ and we write $\bp(i)$ to denote the value of its $i$th coordinate; this notation extends naturally to the set $S \subseteq \cI$ of coordinates, for which we let $\bp(S) := \sum_{i \in S} \bp(i)$. Furthermore, we use $\bp|_S$ to denote the restriction of $\bp$ to $S$; more formally, $\bp|_S(i) = \bp(i)$ if $i \in S$ and $\bp|_S(i) = 0$ otherwise.
We also write $\bp|_{\bar{S}}$ as a shorthand for $\bp - \bp|_S$, i.e., the restriction of $\bp$ to the complement of $S$.
We use $\supp(\bp)$ to denote the set of non-zero coordinates of vector $\bp$. A vector is said to be \emph{$k$-sparse} if its support is of size at most $k$. 
Recall that the $\ell_1$-norm of a vector $\bp \in \R^\cI$ is $\|\bp\|_1 := \sum_{i \in \cI} |p(i)|$. 

\subsection{Earth Mover's Distance (EMD)} 
\label{sec:prelim-emd}

Given two non-negative vectors $\bp, \bq \in \R_{\geq 0}^{G_\Delta}$ such that $\|\bp\|_1 = \|\bq\|_1$, their \emph{Earth Mover's Distance} (EMD) is
\begin{align*}
\textstyle
\emd(\bp, \bq) := \min_{\gamma} \sum_{x \in G_\Delta} \sum_{y \in G_\Delta} \gamma(x, y) \cdot \|x - y\|_1,
\end{align*}
where the minimum is over $\gamma \in \R_{\geq 0}^{G_\Delta \times G_\Delta}$ whose marginals are $\bp$ and $\bq$. (I.e., for all $x \in G_{\Delta}$, $\sum_{y \in G_\Delta} \gamma(x, y) = \bp(x)$ and, for all $y \in G_{\Delta}$, $\sum_{x \in G_\Delta} \gamma(x, y) = \bq(y)$.)


We define the \emph{EMD norm} of a vector $\bw \in \R^{G_\Delta}$ by
\[
\textstyle
\|\bw\|_{\emd} := \min_{\bp, \bq \in \R_{\geq 0}^{G_{\Delta}} \atop \bp - \bq + \br = \bw, \|\bp\|_1 = \|\bq\|_1} \emd(\bp, \bq) + \alpha \cdot \|\br\|,
\]
where $\alpha = 2$ is the diameter of our space $[0, 1) \times [0, 1)$.
%


The following simple lemma will be useful when dealing with unnormalized vs normalized vectors.
\begin{lemma} \label{lem:err-from-normalization}
Suppose that $\bs, \hbs \in \R^{G_\Delta}_{\geq 0}$ are such that $\|\bs\|_1 = n$ and $\|\bs - \hbs\|_{\emd} \leq n / 2$. Let $\ba = \bs / \|\bs\|_1$ and $\hba = \hbs / \|\hbs\|_1$. Then, we have $\|\ba - \hba\|_{\emd} \leq 4 \|\bs - \hbs\|_{\emd} / n$.
\end{lemma}

\begin{proof}
Let $\zeta = \|\bs - \hbs\|_{\emd}$; observe that $|\|\bs\|_1 - \|\hbs\|_1| \geq \zeta$. As a result, we have $\|\hbs\|_1 \in [n - \zeta, n + \zeta]$. Thus,
\begin{align*}
\|\hbs/n - \hba\|_{\emd} 
&\leq \|\hbs\|_{\emd} \cdot \left|\frac{1}{n} - \frac{1}{\|\hbs\|_1}\right| \\
&\leq (n + \zeta) \cdot \left|\frac{1}{n} - \frac{1}{n - \zeta}\right| 
\leq \frac{3\zeta}{n}. 
\end{align*}

As a result, from the triangle inequality, we have
\begin{align*}
\|\ba - \hba\|_{\emd}
& \leq \|\ba - \hbs/n\|_{\emd} + \|\hbs/n - \hba\|_{\emd} \\
& \leq \frac{\zeta}{n} + \frac{3\zeta}{n} = \frac{4\zeta}{n}. \qedhere 
\end{align*}
\end{proof}

\subsection{Differential Privacy}

Two input datasets $\bX, \bX'$ are \emph{neighbors} if $\bX'$ results from adding or removing a single user's data from $\bX$. In our setting, each user $i$'s data is a distribution $\bp_i$ over $G_\Delta$.

\begin{defn}[Differential Privacy;~\citet{dwork2006calibrating}]
A mechanism $\cM$ is said to be \emph{$\eps$-DP} iff, for every set $O$ of outputs and every pair $\bX, \bX'$ of neighboring datasets, $\Pr[\cM(\bX) \in O] \leq e^{\eps} \cdot \Pr[\cM(\bX') \in O].$
\end{defn}

For a vector-valued function $f$, its \emph{$\ell_1$-sensitivity}, denoted by $S_1(f)$, is defined as $\max_{\text{neighbors } \bX, \bX'} \|f(\bX) - f(\bX')\|_1$.
\begin{defn}[Laplace Mechanism]
The \emph{Laplace mechanism} with parameter $b > 0$ adds an independent noise drawn from the Laplace distribution $\Lap(b)$ to each coordinate of a vector-valued function $f$. 
\end{defn}

\begin{lemma}[\citet{dwork2006calibrating}] \label{lem:laplace-dp}
The Laplace mechanism with parameter $S_1(f) / \eps$ is $\eps$-DP.
\end{lemma}

\subsection{Heatmaps}
\label{sec:prelim-heatmap}

Given $\bp \in \R_{\geq 0}^{G_\Delta}$, its associated \emph{heatmap} with Gaussian filter variance $\sigma^2$ is defined as
\begin{align*}
\textstyle
H_\bp^\sigma(x, y) = \sum_{(x', y') \in G_\Delta} \frac{1}{Z(x', y')} e^{-\frac{(x - x')^2 + (y - y')^2}{2\sigma^2}} \cdot p(x', y')
\end{align*}
for all $(x, y) \in G_\Delta$, where $Z(x', y') := \sum_{(x', y') \in G_\Delta} e^{-\frac{(x - x')^2 + (y - y')^2}{2\sigma^2}}$ is the normalization factor.

In the heatmap aggregation problem over $n$ users, each user $i$ has a probability distribution $\bp_i$ over $G_\Delta$. The goal is to output an estimate of the aggregated heatmap $H^\sigma_\ba$ where $\ba = \frac{1}{n} \sum_{i \in [n]} \bp_i$.

\section{Algorithm}
\label{sec:algo-main}

In this section, we describe our private sparse EMD aggregation algorithm and prove our main result.

\begin{theorem} \label{thm:main-ub}
For any $\eps > 0$ and $\lambda \in (0, 1)$, there is an $\eps$-DP algorithm that can, w.p. $0.99$, output a $\left(\lambda, O\left(\frac{\sqrt{k}}{\lambda \eps n}\right)\right)$-approximation for the $k$-sparse EMD aggregation problem.
\end{theorem}

\subsection{Pyramidal Transform}

As alluded to in \Cref{sec:intro}, we use a linear transformation from~\cite{indyk2011k}. This linear transformation is the so-called \emph{(scaled) pyramidal transform}, whose variant is also often used in (metric) embedding of $\emd$ to $\ell_1$~\cite{Charikar02,indyk2003fast}. Roughly speaking, the transform represents a hierarchical partitioning of $[0, 1)^2$ into subgrids, where a subgrid at a level is divided into four equal subgrids at the next level. The (scaled) pyramidal transform has one row corresponding to each subgrid; the row is equal to the indicator vector of the subgrid scaled by its side length. These are formalized below.

\begin{defn}
For $i \in \N \cup \{0\}$, we let $C_{2^i}$ denote the set of \emph{level $i$ grid cells} defined as
$C_{2^i} := \{[a, a + 2^{-i}) \times [b, b + 2^{-i}) \mid (a, b) \in G_{2^i}\}$; let $m_i := |C_{2^i}|$.

For $i \in [\ell]$, the \emph{level-$i$ grid partition map} is defined as the matrix $\bP_i\in \{0, 1\}^{C_{2^{i}} \times G_\Delta}$ where $\bP_i(c, p) = 1$ iff $p \in c$.
The \emph{(scaled) pyramidal transform} is the matrix $\bP \in \R^{\bigcup_{i=0}^\ell C_{2^i} \times G_\Delta}$ defined by
\ifarxiv
$$
\else
$
\fi
\bP := 
\begin{bmatrix}
\bP_0^\top\:2^{-1}\bP_1^\top \hdots2^{-\ell} \bP_\ell^\top
\end{bmatrix}^\top
.
\ifarxiv
$$
\else
$
\fi
\end{defn}

\subsection{The Algorithm}

Our algorithm for sparse EMD aggregation consists of two components.  The first component (\Cref{alg:dp-aggregation-full})  aggregates the input distributions (Line~\ref{line:alg1-agg}) and applies the pyramidal transform to the aggregate, adding different amounts of Laplace noise for different levels of the grid (Lines~\ref{line:alg1-lap}, \ref{line:alg1-noise}). (The parameters $\eps_1, \dots, \eps_\ell$, which govern the amount of Laplace noise, will be specified in the next subsection.) 
The second component (\Cref{alg:dp-aggregation-reconstruction}) takes these noisy measurements for every level of the grid and reconstructs the solution by first recovering the $\ell_1$ solution (Line~\ref{line:l1-recovered}) and then the EMD solution using a linear program (Line~\ref{line:lp}).  

We stress that our algorithm is similar to that of~\citet{indyk2011k} except for two points: first, we add noise to the measurements and, second, we are not doing any ``compression'' in contrast to~\cite{indyk2011k}, which takes a wide matrix $\bA$ for $\ell_1$ recovery and multiplies it with $\bP\bs$.

\begin{figure}
\begin{minipage}{0.47\textwidth}
\begin{algorithm}[H]
\caption{\textsc{DPSparseEMDAgg}}
\label{alg:dp-aggregation-full}
\begin{algorithmic}[1]
\STATE \textbf{Input: } distributions $\bp_1, \dots, \bp_n$ on $G_{\Delta}$
\STATE \textbf{Parameters: } $\eps_1, \dots, \eps_\ell > 0, w \in \N$
\STATE $\bs \leftarrow \sum_{i=1}^n \bp_i$
\label{line:alg1-agg}
\FOR{$i = 0, \dots, \ell$}
\STATE $\bnu_i \leftarrow \Lap(1/\eps_i)^{\otimes m_i}$
\label{line:alg1-lap}
\STATE $\by'_i \leftarrow \frac{1}{2^i}\left(\bP_i\bs + \bnu_i\right)$
\label{line:alg1-noise}
\ENDFOR
\STATE $\by' \leftarrow [\by'_0 \cdots \by'_\ell]$
\STATE $\hbs \leftarrow$ \textsc{Reconstruct}($\by'; w$) 
\RETURN $\hba := \hbs / \|\hbs\|$
\end{algorithmic}
\end{algorithm}
\end{minipage}
\hfill
\begin{minipage}{0.47\textwidth}
\small
\begin{algorithm}[H]
\caption{\textsc{Reconstruct}}
\label{alg:dp-aggregation-reconstruction}
\begin{algorithmic}[1]
\STATE \textbf{Input: } noisy measurements $\by' \in \R^{\bigcup_{i \in [\ell]} C_{2^i}}$
\STATE \textbf{Parameters: } $w \in \N$
\STATE $S_0 \leftarrow C_1$
\FOR{$i = 1, \dots, \ell$}
\STATE $T_i \leftarrow \children(S_{i - 1})$
\STATE $S_i \leftarrow$ the set of $\min\{w, |T_i|\}$ coordinates in $T_i$ with maximum values in $\by'$
\ENDFOR
\STATE $S \leftarrow \bigcup_{i \in [\ell]} S_i$
 and $\hby \leftarrow \by'|_S$ \label{line:l1-recovered}
\RETURN $\hbs \leftarrow \argmin_{\bs' \geq \bzero} \|\hby - \bP\bs'\|_1$
\label{line:lp}
\end{algorithmic}
\end{algorithm}
\end{minipage}
\end{figure}

\subsection{Analysis}

Following the framework of~\citet{indyk2011k}, our analysis proceeds in two stages.  We first show that the ``recovered'' $\hby$ is close, in the $\ell_1$ metric, to the true value of $\bP\bs$. Then, we use the properties of $\bP$ to argue that the output $\hbs$ is close, in EMD, to $\bs$. Since we are adding noise to our measurement, we need to extend the work of~\citet{indyk2011k} to be robust to noise. Finally, we set the privacy parameters $\eps_1, \dots, \eps_\ell$ 
to finish our proof of \Cref{thm:main-ub}.

Let us now briefly demystify the additive error bound $O_{\eps, \lambda}(\sqrt{k})$ that we end up with for $\hbs$ (which ultimately gives the $O_{\eps, \lambda}(\sqrt{k} / n)$ error bound for the normalized $\hba$). We will select $w = O_{\lambda}(k)$ so as to have an additive error of $O_{\eps}(\sqrt{w})$. At a high level, each noise $\frac{1}{2^i} \cdot \bnu_i(t)$ added to a ``queried'' term $\by_i(t)$ for $t \in T_i$ permeates to an error of the same order. For simplicity, assume for the moment that $|\bnu_i(t)| = O(1/\eps_i)$. Now, notice that if we are at level $i < \log \sqrt{w}$, then $|T_i| = |C_{2^i}| = 2^{2i}$ and thus the total error contribution of this level is $O(2^i /\eps_i)$. On the other hand, for a level $i \geq \log\sqrt{w}$, we will have $|T_i| = w$ and the error contribution is $O\left(\frac{w}{2^i \eps_i}\right)$. Now, when $i = \log \sqrt{w} \pm O(1)$, these error terms are $O(\sqrt{w} / \eps_i)$ and thus we should set $\eps_i = \Omega(1)$ to get the desired bound. However, in terms of $|i - \log \sqrt{w}|$, these error terms become exponentially smaller, i.e., $O\left(\frac{\sqrt{w}}{2^{|i - \log \sqrt{w}|} \eps_i}\right)$. This leads to the natural choices of $\eps_i$ we use, which is to make it proportional to $\gamma^{|i - \log \sqrt{w}|}$ for some constant $\gamma > 0.5$. This indeed leads to the desired $O_{\eps}(\sqrt{w}) = O_{\eps, \lambda}(\sqrt{k})$ bound.

\paragraph{Phase I: $\ell_1$ Recovery.}
We will now analyze the $\ell_1$ recovery guarantee of $\hby$. Our recovery algorithm, which is an adaptation of~\citet{indyk2011k}, does \emph{not} work for general hidden vectors. However, it works well for those that follow a certain ``tree-like structure'', formalized below.

\begin{defn}[\citet{indyk2011k}]
For $i \geq 1$, a grid cell $c' \in C_{2^i}$ is said to be a \emph{child} of grid cell $c \in C_{2^{i - 1}}$ if $c \subseteq c'$. This forms a tree rooted at $[0, 1) \times [0, 1) \in C_0$ where every internal node has exactly four children. We let $\cT_w$ denote the set of all trees 
such that the number of nodes at each level is at most $w$. 

Let $\cM_{w}$ denote the set of $\by = [\by_0 \cdots \by_\ell]$ where $\by_i \in \R_{\geq 0}^{C_{2^i}}$ such that 
\begin{enumerate}
\item $\supp(\by) \subseteq T$ for some tree $T \in \cT_w$.
\item For all $i \in [\ell - 1], p \in C_{2^i}$, the following holds: $\by(p) \geq 2 \cdot \by(\children(p))$.
\end{enumerate}
\end{defn}

Under the above notion, we can adapt the $\ell_1$ recovery analysis of~\citet{indyk2011k} in the no-noise case to our regime, where the noise shows up as an error:

\begin{lemma} \label{lem:l1-recovery-tree}
Let $\by^* \in \argmin_{\by \in \cM_w} \|\bP\bs - \by\|_1$ where $\supp(\by^*) \subseteq T^*$ for some $T^* \in \cT_w$; let $T^*_i$ denote $T^* \cap C_{2^i}$ and $V_i = T^*_i \setminus S_i$ for all $i \in [\ell]$. Then, $\hby$ on Line~\ref{line:l1-recovered} of \textsc{Reconstruct} satisfies
$
\textstyle
\|\hby - \bP\bs\|_1 \leq 3\|\by^* - \bP\bs\|_1 + O\left(\sum_{i \in [\ell]} \frac{1}{2^i} \left\|\bnu_i|_{V_i \cup S_i}\right\|_1\right).
$
\end{lemma}

\begin{proof}
For every $q \in T^* \setminus S$, let $R(q)$ be the highest ancestor of $q$ that does not belong to $S$. We have
\begin{align}
\|\by^*|_{\bar{S}}\|_1 
&= \sum_{q \in T^* \setminus S} \by^*(q) 
= \sum_{i \in [\ell]} \sum_{p \in V_i} \sum_{q \in R^{-1}(p)} \by^*(q) \nonumber \\
& \overset{(\diamondsuit)}{\leq} \sum_{i \in [\ell]} \sum_{p \in V_i} 2\by^*(p) 
= 2 \sum_{i \in [\ell]} \by^*(V_i), \label{eq:expand-decay}
\end{align}
where $(\diamondsuit)$ follows from the second property of $\cM_w$.

Next, consider the algorithm at the $i$th iteration and $p \in V_i$. Since $p$ was not picked, the following must hold for all $q \in S_i \setminus T^*_i$:
$\by'(p) \leq \by'(q)$.
Observe also that from $|S_i| = \max\{w, |C_{2^i}|\}$ and $|T^*_i| \leq \max\{w, |C_{2^i}|\}$, we also have $|S_i \setminus T^*_i| \geq |T^*_i \setminus S_i| = |V_i|$. Thus, we get
\begin{align} \label{eq:dropped-each-level}
\by'(V_i) \leq \by'(S_i \setminus T^*_i).
\end{align}

From this and~\eqref{eq:expand-decay}, we can further derive
\begin{align}
&\|\by^*|_{\bar{S}}\|_1 \nonumber \\
&\overset{\eqref{eq:expand-decay}}{\leq} 2 \left(\sum_{i \in [\ell]} (\by^*(V_i) - \bP\bs(V_i)) + \bP\bs(S_i \setminus T^*_i) \right) \nonumber \\ &\qquad 
+ \left(\sum_{i \in [\ell]} \bP\bs(V_i) - \bP\bs(\left(S_i \setminus T^*_i\right))\right) \nonumber \\
&\overset{(\square)}{\leq} 2\|\by^* - \bP\bs\|_1 + 2 \left(\sum_{i \in [\ell]} \bP\bs(V_i) - \bP\bs(\left(S_i \setminus T^*_i\right))\right) \nonumber \\
&\overset{(\bigtriangleup)}{\leq} 2\|\by^* - \bP\bs\|_1 + 2 \left(\sum_{i \in [\ell]} \frac{1}{2^i} \left\|\bnu_i|_{V_i \cup \left(S_i \setminus T^*_i\right)}\right\|_1\right) \nonumber \\ &\qquad
+ 2\left(\sum_{i \in [\ell]} \by'(V_i) - \by'(\left(S_i \setminus T^*_i\right))\right) \nonumber \\
& \overset{\eqref{eq:dropped-each-level}}{\leq} 2\|\by^* - \bP\bs\|_1 + 2 \left(\sum_{i \in [\ell]} \frac{1}{2^i} \left\|\bnu_i|_{V_i \cup \left(S_i \setminus T^*_i\right)}\right\|_1\right), \label{eq:optimal-norm-outside} 
\end{align}
where $(\square)$ follows from $\supp(\by^*) \subseteq T^* = \bigcup_{i \in [\ell]} T^*_i$ and $(\bigtriangleup)$ follows from how $\by'$ is calculated.
%
Finally, from $\hby = \by'|_S$ and how each entry of $\by'$ is computed, we have
\begin{align*}
& \|\hby - \bP\bs\|_1 
= \|\by'|_{S} - \bP\bs|_{S}\|_1 + \|\bP\bs_{\bar{S}}\|_1 \\
&\leq \left(\sum_{i \in [\ell]} \frac{1}{2^i} \|\bnu_i|_{S_i}\|_1\right) + \|\by^*|_{\bar{S}}\|_1 + \|\by^*|_{\bar{S}} - \bP\bs|_{\bar{S}}\|_1 \\
&\overset{\eqref{eq:optimal-norm-outside}}{\leq} \left(\sum_{i \in [\ell]} \frac{1}{2^i} \|\bnu_i|_{S_i}\|_1\right) \\
& \qquad + \left(2\|\by^* - \bP\bs\|_1 + 2 \left(\sum_{i \in [\ell]} \frac{1}{2^i} \left\|\bnu_i|_{V_i \cup \left(S_i \setminus T^*_i\right)}\right\|_1\right)\right) \\
& \qquad + \|\by^* - \bP\bs\|_1 \\
&\leq 3\|\by^* - \bP\bs\|_1 + 3\left(\sum_{i \in [\ell]} \frac{1}{2^i} \left\|\bnu_i|_{V_i \cup S_i}\right\|_1\right). \qedhere
\end{align*}
\end{proof}

\paragraph{Phase II: From $\ell_1$ to EMD.}
We now proceed to bound the EMD error. The main lemma is stated below.
%
\begin{lemma} \label{lem:final-guarantee-deterministic}
Let the notation be as in \Cref{lem:l1-recovery-tree}. For any $\eta' \in (0, 1)$, by setting $w = O(k/(\eta')^2)$, the output $\hbs$ of \textsc{Reconstruct} satisfies
$
\textstyle
\|\bs - \bs^*\|_{\emd} \leq \eta' \cdot \min_{k\text{-sparse } \bs'} \|\bs - \bs'\|_{\emd} + O\left(\sum_{i \in [\ell]} \frac{1}{2^i} \left\|\bnu_i|_{V_i \cup S_i}\right\|_1\right).
$
\end{lemma}

Similar to the proof of~\citet{indyk2011k}, our proof of \Cref{lem:final-guarantee-deterministic} converts the recovery guarantee under $\ell_1$ metric to that under EMD; to do this, we need the following two statements from prior work. 
\begin{lemma}[Model-Alignment of EMD with $\cM_w$~\cite{indyk2011k}] \label{lem:model-alignment}
For any $\bx \in \R_{\geq 0}^{G_\Delta}, k \in \N$ and $\eta \in (0, 1)$, there exist $w = O(k / \eta^2)$ and $\by^* \in \cM_w$ such that $\|\by^* - \bP\bs\|_1 \leq \eta \cdot \min_{k\text{-sparse } \bx'} \|\bx - \bx'\|_{\emd}.$
\end{lemma}

\begin{lemma}[EMD-to-$\ell_1$ Expansion~\cite{indyk2003fast}] \label{lem:emd-to-l1-expansion}
For all $\bz \in \R^{G_{\Delta}}$, $\|\bz\|_{\emd} \leq \|\bP\bz\|_1.$
\end{lemma}

\begin{proof}[Proof of \Cref{lem:final-guarantee-deterministic}]
Recall that we use $\bs$ to denote the true sum $\sum_{i=1}^n \bp_i$. We set $\eta = \eta' / 6$ and let $w = O(k/\eta^2) = O(k/(\eta')^2)$ be as in \Cref{lem:model-alignment}, which ensures that there exists $\by^* \in \cM_w$ with
\begin{align} \label{eq:tmp1}
\|\by^* - \bP\bs\|_1 \leq \eta \cdot \min_{k\text{-sparse } \bs'} \|\bs - \bs'\|_{\emd}.
\end{align}
Thus, using \Cref{lem:emd-to-l1-expansion}, we can derive
\begin{align*}
\|\bs - \bs^*\|_{\emd} &\leq \|\bP(\bs - \bs^*)\|_1 \\
(\text{triangle inequality}) &\leq \|\hby - \bP\bs\|_1 + \|\hby - \bP\bs^*\|_1 \\
(\text{how } \bs^* \text{ is computed}) &\leq 2\|\hby - \bP\bs\|_1 \\
(\text{\Cref{lem:l1-recovery-tree}}) &\leq 6\|\by^* - \bP\bs\|_1 \\
& \quad + O\left(\sum_{i \in [\ell]} \frac{1}{2^i} \left\|\bnu_i|_{V_i \cup S_i}\right\|_1\right) \\
&\overset{\eqref{eq:tmp1}}{\leq} \eta' \cdot \min_{k\text{-sparse } \bs'} \|\bs - \bs'\|_{\emd} \\
& \quad + O\left(\sum_{i \in [\ell]} \frac{1}{2^i} \left\|\bnu_i|_{V_i \cup S_i}\right\|_1\right). \qedhere
\end{align*}
\end{proof}

\begin{figure*}[h!]
\centering
\includegraphics[trim={5cm 0 0 0},clip,scale=0.26]{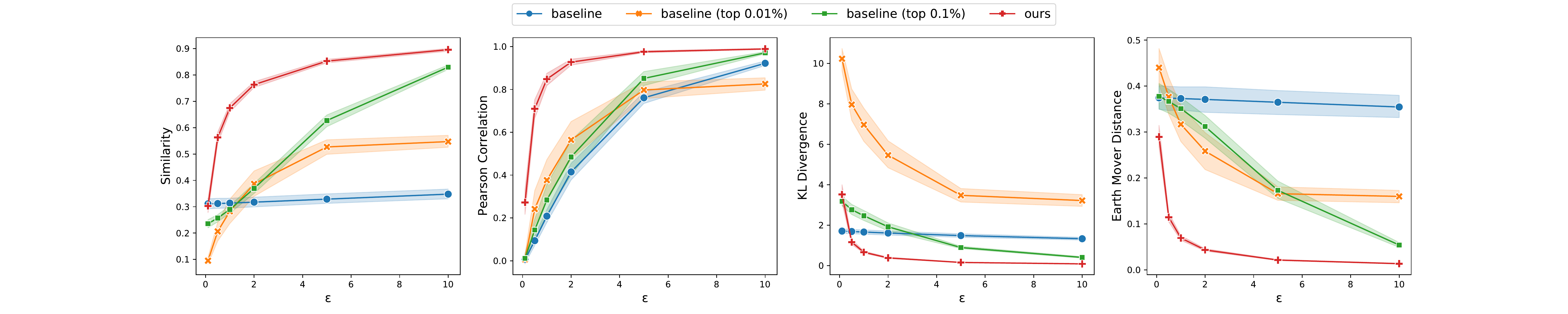}
\caption{Metrics averaged over 60 runs when varying $\eps$. Shaded areas indicate 95\% confidence interval.
\label{fig:comparison-result}}
\end{figure*}

\begin{figure}[h!]
\centering
\includegraphics[width=0.4\textwidth]{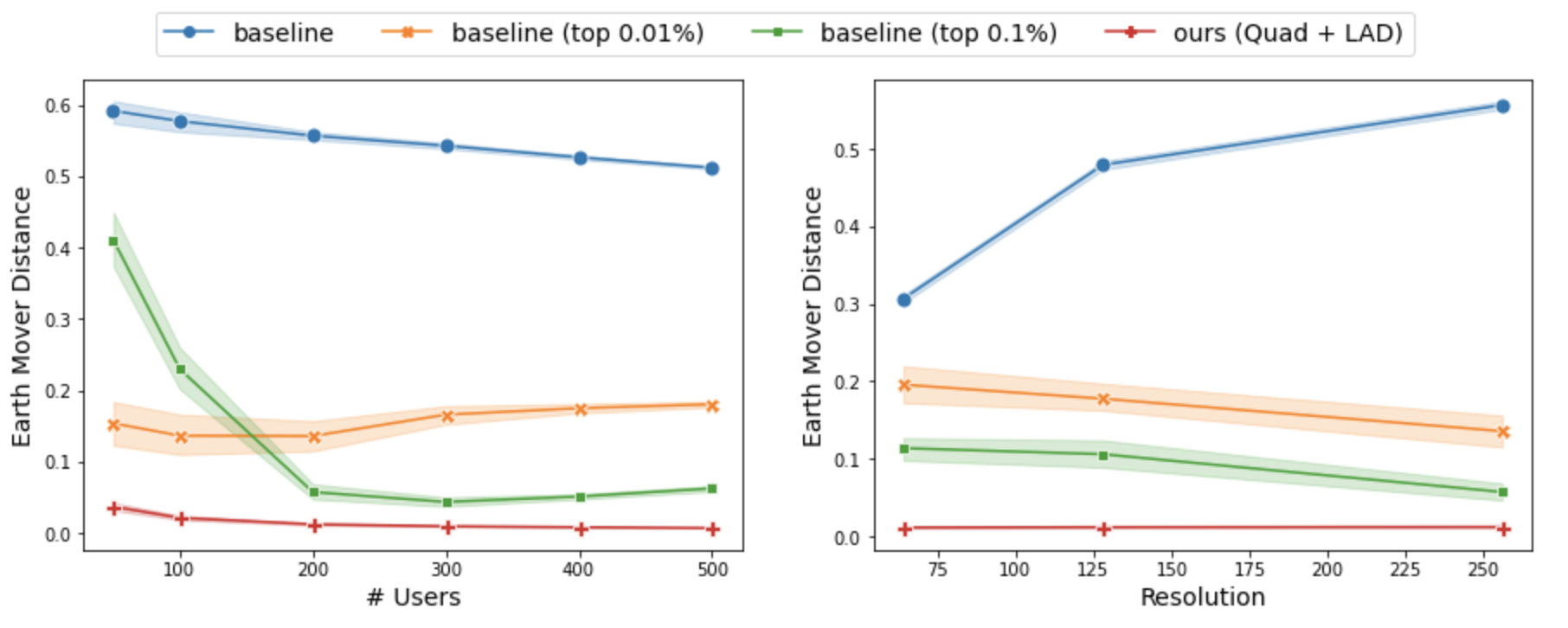}
\caption{Effect of the \#users and grid resolution on  EMD. 
\label{fig:user-resolution}}
\end{figure}

\paragraph{Finishing the Proof.}
We now select the privacy parameters and complete the proof of \Cref{thm:main-ub}.
\begin{proof}[Proof of \Cref{thm:main-ub}]
Let $w = O(k/(\eta')^2)$ be as in \Cref{lem:final-guarantee-deterministic} with $\eta' = \lambda / 4$, and let $q = \lfloor \log_2 \sqrt{w} \rfloor$. Let $\gamma = 0.8$ be the ``decay rate'' for $\eps_i$'s, and let $Z = \sum_{i=0}^{\ell} \gamma^{|i - q|} \leq O(1)$ be the normalization factor. We run \Cref{alg:dp-aggregation-full} with $\eps_i = \gamma^{|i - q|} \cdot \eps / Z$. 

\paragraph{Privacy Analysis.} We can view the $i$th iteration of the algorithm as releasing $2^i \by'_i = \bP_i\bs + \bnu_i$. Since each $\bp_i$ has $\ell_1$ norm at most one, its sensitivity with respect to $\bP_i\bs$ is at most one; thus, \Cref{lem:laplace-dp} implies that the $i$th iteration is $\eps_i$-DP. As a result, by basic composition theorem of DP, we can conclude that releasing all of $\by'_0, \dots, \by'_\ell$ is $(\eps_0 + \cdots + \eps_\ell)$-DP. Since the reconstruction is simply a post-processing step, the post-processing property of DP ensures that \Cref{alg:dp-aggregation-full} is $(\eps_0 + \cdots + \eps_\ell)$-DP. Finally, observe that by definition of $\eps_i$'s, we have $\eps_0 + \cdots + \eps_\ell = \eps$ as desired.

\paragraph{Utility Analysis.} Applying \Cref{lem:final-guarantee-deterministic}, we can conclude that $\|\bs - \bs^*\|_{\emd} \leq \eta' \cdot \min_{k\text{-sparse } \bs'} \|\bs - \bs'\|_{\emd} + \xi$,
where $\xi = O\left(\sum_{i \in [\ell]} \frac{1}{2^i}\left\|\bnu_i|_{V_i \cup S_i}\right\|_1\right)$.
Recall that each of $V_i, S_i$'s is of size at most $\max\{w, 2^{2i}\}$ (because of definition of $\cM_w$ and the fact that $m_i = |C_{2^i}| = 2^{2i}$), and that each entry of $\bnu_i$ is sampled from $\Lap(1/\eps_i)$.  As a result, we have
\begin{align*}
\textstyle
\E[\xi] & \leq O\left(\sum_{i \in [\ell]} \frac{1}{2^i} \cdot \max\{w, 2^{2i}\} \cdot \frac{1}{\eps_i} \right) \\
&= O\left(\sum_{i \in [q]} \frac{2^i}{\gamma^{q - i} \eps}\right) + O\left(\sum_{i \in \{q + 1, \dots, \ell\}} \frac{k}{\lambda^2} \cdot \frac{1}{2^i \gamma^{i - q} \eps}\right)  \\
& = O\left(\frac{2^q}{\eps}\right) + O\left(\frac{k}{\lambda^2} \cdot \frac{1}{2^q} \cdot \frac{1}{\eps}\right)
= O\left(\frac{\sqrt{k}}{\lambda \eps}\right),
\end{align*}
where the last bound follows from our choice of $\gamma > 0.5$ and $q = \lfloor \log_2 \sqrt{w} \rfloor$.
Hence, by Markov's inequality, w.p. 0.99, we have
$
\textstyle
\|\bs - \bs^*\|_{\emd} \leq \eta' \cdot \min_{k\text{-sparse } \bs'} \|\bs - \bs'\|_{\emd} + 100\E[\xi] = \eta' \cdot \min_{k\text{-sparse } \bs'} \|\bs - \bs'\|_{\emd} + O\left(\frac{\sqrt{k}}{\lambda \eps}\right).
$
Finally, applying \Cref{lem:err-from-normalization} concludes the proof.
\end{proof}

\begin{figure*}[h!]
\centering
\includegraphics[trim={7.5cm 0cm 7.5cm 0.9cm},clip,scale=0.21]{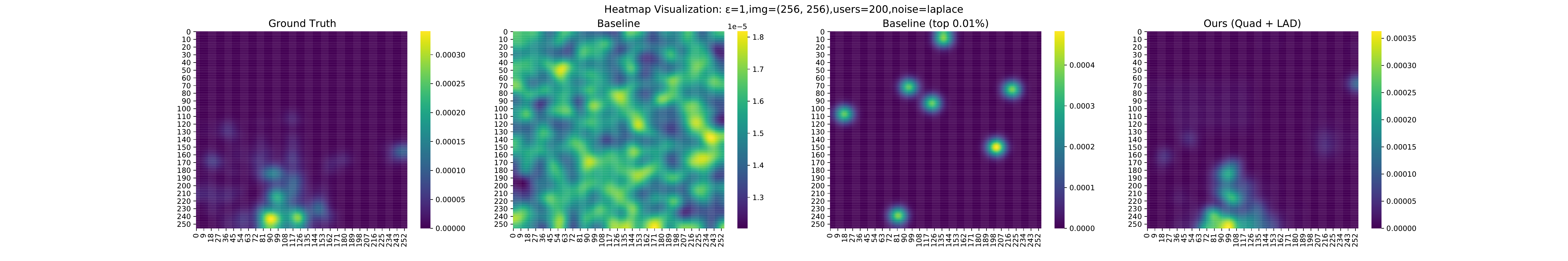}
\includegraphics[trim={7.5cm 0cm 7.5cm 0.9cm},clip,scale=0.21]{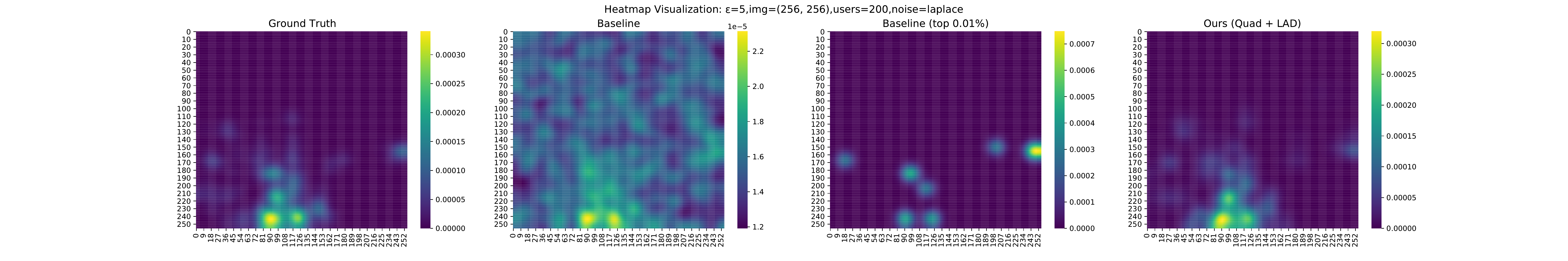}
\caption{Example visualization of different algorithms for $\eps=1$ (top) and $\eps=5$ (bottom). The algorithms from left to right are: original heatmap (no privacy), baseline, baseline with top 0.01\% and our algorithm. \label{fig:visualization}}
\end{figure*}

\section{Experiments}
\label{sec:exp}

In this section, we study the performance of our algorithms on real-world datasets. 
\ifarxiv
Additional experiments on other datasets, including the Salicon image saliency dataset~\cite{jiang2015salicon} and a synthetic dataset, can be found in Appendix~\ref{app:exp}.
\fi

\newcommand{\gowalla}{\textsc{Gowalla}\xspace}
\newcommand{\brightkite}{\textsc{Brightkite}\xspace}

\paragraph{Implementation Details.}
We implement \Cref{alg:dp-aggregation-full} with a minor modification: we do not measure at the level $i < q = \lfloor \log\sqrt{w} \rfloor$. In other words, we start directly at the lowest level for which the number of grid cells is at most $\sqrt{w}$. It is possible to adjust the proof to show that, even with this modification, the error remains $O_\eps(\sqrt{k})$. Apart from this, the algorithm is exactly the same as presented earlier. We note that the linear program at the end of \Cref{alg:dp-aggregation-reconstruction} can be formulated so that the number of variables is only $O(w \ell)$; the reason is that we only need one variable per cell that is left out at each stage. This allows us to solve it efficiently even when the resolution $\Delta = 2^\ell$ is large.

As for our parameters, we use the decay rate $\gamma = 1/\sqrt{2}$, which is obtained from minimizing the second error term in the proof of \Cref{thm:main-ub} as $\ell \to \infty$\footnote{When $w$ (and thus $q$) is fixed, the second error term is proportional to $Z \cdot \sum_{i=0}^{\ell - q - 1} \frac{1}{(2\gamma)^i}$ which converges to $\frac{1}{(1 - \gamma)(1 - 0.5/\gamma)}$ as $\ell \to \infty$. The latter term is minimized when $\gamma = 1/\sqrt{2}$}. We use $w = 20$ in our experiments, which turns out to work well already for datasets we consider. We refrain from tuning parameters further since a privacy analysis of the tuning step has to be taken into account if we want to be completely rigorous. (See, e.g.,~\cite{LT19} for a formal treatment.)

\paragraph{Datasets.}
We use two datasets available at \url{snap.stanford.edu} to generate the input distribution for users.  The first dataset%
\footnote{Available at \url{http://snap.stanford.edu/data/loc-Gowalla.html}}, called \gowalla, consists of location check-ins by users of the location-based social network Gowalla.  Each record consists of, among other things, an anonymized user id together with the latitude (lat) and longitude (lon) of the check-in and a timestamp.  We filtered this dataset to consider only check-ins roughly in the continental US (i.e., lon $\in (-135, -60)$ and lat $\in (0, 50)$) for the month of January 2010; this resulted in 196,071 check-ins corresponding to 10,196 users.  
The second dataset%
\footnote{Available at \url{http://snap.stanford.edu/data/loc-Brightkite.html}}, called \brightkite, also contains check-ins from a different and now defunct location-based social network Brightkite; each record is similar to \gowalla.  Once again, we filtered this dataset to consider only check-ins in the continental US for the months of November and December 2008; this resulted in 304,608 check-ins corresponding to 10,177 users.  

For each of these datasets, we partition the whole area into a $300 \times 300$ grid. We then took the top 30 cells (in both datasets combined) that have the most check-ins.  (Each of the 30 cells is mostly around some city like New York, Austin, etc, and has check-ins from at least 200 unique users).  We then consider each cell, partition into $\Delta \times \Delta$ subgrids and snap each check-in to one of these subgrids.  

\paragraph{Metrics.}
To evaluate the quality of an output heatmap $\hbh$ compared to the true heatmap $\bh$, we use the following commonly used metrics: Similarity, Pearson coefficient, KL-divergence, and EMD. (See, e.g.,~\cite{BylinskiiJOTD19} for detailed discussions of these metrics.) We note that the first two metrics should \emph{increase} as $\hbh, \bh$ are more similar, whereas the latter two should \emph{decrease}.

\paragraph{Baselines.}
We consider as a baseline an algorithm recently proposed in~\cite{liu2019differential},\footnote{We remark that~\cite{liu2019differential} also propose using the Gaussian mechanism. However, this algorithm does \emph{not} satisfy $\eps$-DP. Moreover, even when considering $(\eps, \delta)$-DP for moderate value of $\delta$ (e.g., $\delta = 10^{-3}$), the Gaussian mechanism will still add more noise in expectation than the Laplace mechanism.} where we simply add Laplace noise to each subgrid cell of the sum $\bs$, zero out any negative cells, and produce the heatmap from this noisy aggregate. We also consider a ``thresholding'' variant of this baseline that is more suited to sparse data: only keep top $t\%$ of the cell values after noising (and zero out the rest). 

\subsection{Results}
\label{subsec:exp-results-main}
In the first set of experiments, we fix $\Delta = 256$.
For each $\eps \in \{0.1, 0.5, 1, 2, 5, 10\}$, we run our algorithms together with the baseline and its variants on all 30 cells, with 2 trials for each cell. In each trial, we sample a set of 200 users and run all the algorithms; we then compute the distance metrics between the true heatmap and the estimated heatmap. The average of these metrics over the 60 runs is presented in~\Cref{fig:comparison-result}, together with the 95\% confidence interval. As can be seen in the figure, the baseline has rather poor performance across all metrics, even for large $\eps = 10$. We experiment with several values of $t$ for the thresholding variant, which yields a significant improvement. Despite this, we still  observe an advantage of our algorithm consistently across all metrics. These improvements are especially significant when $\eps$ is not too large or too small (i.e., $0.2 \leq \eps \leq 5$).

In the second set of experiments, we study the effect of varying the number $n$ of users. By fixing a single cell (with $>$ 500 users) and $\eps$, we sweep $n \in \{50, 100, 200, 300, 400, 500\}$ users. For each value of $n$, we run 10 trials and average their results. As predicted by theory, our algorithms and the original baseline perform better as $n$ increases. However, the behavior of the thresholding variants of the baseline are less predictable, and sometimes the performance degrades with a larger number of users. It seems plausible that a larger number of users cause an increase in the sparsity, which after some point makes the simple thresholding approach unsuited for the data.

We also run another set of experiments 
where we fix a single cell and $\eps$, and vary the resolution $\Delta \in \{64, 128, 256\}$. In agreement with theory, our algorithm's utility remains nearly constant for the entire range of $\Delta$. On the other hand, the original baseline suffers across all metrics as $\Delta$ increases. The thresholding variants are more subtle; they occasionally improve as $\Delta$ increases, which might be attributed to the fact that when $\Delta$ is small, thresholding can zero out too many subgrid cells.

\Cref{fig:user-resolution} presents results from the latter two set of experiments under the EMD metric when $\eps = 10$. 
We also include examples of the heatmaps from each approach in \Cref{fig:visualization}.

\section{Discussions and Future Directions}
\label{sec:conc}

We present an algorithm for sparse distribution aggregation under the EMD metric, which in turn yields an algorithm for producing heatmaps. As discussed earlier, our algorithm extends naturally to distributed models that can implement the Laplace mechanism, including the secure aggregation model and the shuffle model~\cite{balle_merged,ghazi2019private}. Unfortunately, this does not apply to the more stringent \emph{local} DP model~\cite{kasiviswanathan2008what} and it remains an interesting open question to devise practical local DP heatmap/EMD aggregation algorithms for ``moderate'' number of users and privacy parameters.

\ifarxiv
\bibliographystyle{alpha}
\fi
\bibliography{refs}

\ifarxiv
\appendix

\section{Implications for Planar $k$-Median}
\label{app:kmedian}

Sparse EMD aggregation is intimately related to clustering, specifically, the planar $k$-median problem. Indeed, as already noted (somewhat implicitly) by~\citet{indyk2011k,BackursIRW16}, an algorithm for sparse EMD aggregation immediately implies a coreset~\cite{DBLP:conf/stoc/Har-PeledM04} for $k$-median on the plane, which in turn can be used to derive an approximation algorithm for the problem.

Before we can state our results, recall that in the \emph{planar $k$-median} problem, each user has a point $x_i \in [0, 1)^2$. The goal is to output a set $C \subseteq [0, 1)^2$ of $k$ centers that minimizes the following objective:
\begin{align*}
\cost_C(X) := \sum_{x \in X} \min_{c \in C} \|c - x\|_1,
\end{align*}
where $X = (x_1, \dots, x_n)$.

The above cost notation extends naturally to that of a non-negative vector $\bx \in \R^{G_\Delta}_{\geq 0}$. More specifically, we can define
\begin{align*}
\cost_C(\bx) := \sum_{u \in G_\Delta} \bx(u) \cdot \min_{c \in C} \|c - u\|_1.
\end{align*}

\begin{defn}[Coreset]
Let $\lambda \in (0, 1)$ and $\kappa > 0$.
$\bx$ is said to be a \emph{$(\lambda, \kappa)$-coreset} for $X$ (with respect to $k$-median objective) iff, for every $C \subseteq [0, 1)^2$ of size $k$, we have
$(1 - \lambda) \cdot \cost_C(X) - \kappa \leq \cost_C(\bx) \leq (1 + \lambda) \cdot \cost_C(X) + \kappa$.
\end{defn}

In this context, our results immediately imply the following: 
\begin{corollary} \label{cor:coreset}
There exists an $\eps$-DP algorithm that, for any constant $\lambda \in (0, 1)$, can output a $(\lambda, O(\sqrt{k} / \eps))$-coreset for $k$-median, where each input point belongs to $[0, 1)^2$.
\end{corollary}

\begin{corollary} \label{cor:k-median}
There exists an $\eps$-DP algorithm that, for any constant $\lambda \in (0, 1)$, can output a $(1 + \lambda, O(\sqrt{k} / \eps))$-approximation for $k$-median, where each input point belongs to $[0, 1)^2$.
\end{corollary}

There have been numerous works studying DP clustering algorithms~\citep{blum2005practical, nissim2007smooth, feldman2009private, GuptaLMRT10, mohan2012gupt, wang2015differentially, NissimSV16, nock2016k, su2016differentially, feldman2017coresets, BalcanDLMZ17, NissimS18,huang2018optimal,nock2016k,NissimS18, StemmerK18,Stemmer20, GKM20, jones2020differentially, chaturvedi2020differentially,CGKM21}, each with varying guarantees. Many of these works either considered the low-dimensional case directly (e.g.,~\cite{feldman2009private}), or used it as a subroutine (e.g.,~\cite{BalcanDLMZ17,GKM20}). Despite this, each of these works incurs an additive error of $\Omega(k)$, while our results above achieve an improved $O(\sqrt{k})$ additive error; this bound is also tight due to a simple adaptation of our lower bound for $k$-sparse distribution aggregation. Another surprising aspect of our guarantee is that the additive error is \emph{independent} of the number $n$ of users. To the best of our knowledge, all previous results incur an additive error that depends on $n$.

We will now prove \Cref{cor:coreset} and \Cref{cor:k-median}. In both cases, we will assume that each point (together with the centers) belongs to $G_\Delta$, where $\Delta = \lceil \log(\eps / (n\sqrt{k})) \rceil$.  Indeed, if a point was not from $G_\Delta$, then we can simply replace it with the closest point in $G_\Delta$, which would introduce an additive error of at most $O(\sqrt{k}/\eps)$.

Below we use the notation, $\bone_S$ to denote the indicator vector of $S$, i.e.,
\begin{align*}
\bone_S(a) =
\begin{cases}
1 &\text{ if } a \in S, \\
0 &\text{ otherwise.}
\end{cases}
\end{align*}

\begin{proof}[Proof of \Cref{cor:coreset}]
We view the input point $x_i$ as the indicator vector $\bp_i = \bone_{x_i}$. We then run \Cref{alg:dp-aggregation-full} but output $\bx = \hbs$ (instead of the normalized $\hbs / \|\hbs\|_{\emd}$). Recall from the proof of \Cref{thm:main-ub} that, with appropriate setting of parameters, the output satisfies
\begin{align} \label{eq:emd-close-for-coreset}
\|\bs - \hbs\|_{\emd} \leq \lambda \cdot \min_{k\text{-sparse} \bs'} \|\bs - \bs'\|_{\emd} + O(\sqrt{k} / \eps),
\end{align}
where $\bs$ represents $\sum_{i \in [n]} \bp_i$.

Consider any set $C$ of at most $k$ centers. We have
\begin{align*}
&|\cost_C(X) - \cost_C(\bx)| \\
&\leq \|\bs - \hbs\|_{\emd} \\
&\overset{\eqref{eq:emd-close-for-coreset}}{\leq}  \lambda \cdot \min_{k\text{-sparse} \bs'} \|\bs - \bs'\|_{\emd} + O(\sqrt{k} / \eps) \\
&\leq \lambda \cdot \cost_C(X) + O(\sqrt{k} / \eps),
\end{align*}
which concludes our proof.
\end{proof}

\Cref{cor:k-median} then follows from \Cref{cor:coreset} by simply running any non-private algorithm for $k$-median on top of the coreset produced by \Cref{cor:coreset}.
\section{Relationship between Heatmaps Aggregation and EMD Aggregation}
\label{app:rel}

Aggregating distributions under EMD is closely related to aggregating heatmaps. In this section we formalize this relationship by showing that EMD bound on the error of the distribution aggregation problem yields many types of error bounds of aggregating a slight variant of heatmaps, which we call ``infinite'' heatmaps. The latter is the same as the usual heatmaps except that we do not truncate the distribution around the border. More specifically, we let $\tG_\Delta := \{(i/\Delta, j/\Delta) \mid i, j \in \BZ\}$ and 
\begin{align*}
\tH_\bp^\sigma(a) = \sum_{a' \in G_\Delta} \frac{1}{Z} e^{-\frac{\|a - a'\|_2^2}{2\sigma^2}} \cdot p(a'),
\end{align*}
for all $a \in \tG_\Delta$, where $Z := \sum_{d \in \tG_\Delta} e^{-\frac{\|d\|_2^2}{2\sigma^2}}$ is the normalization factor.

We conjecture that (slight variants of) the bounds we discuss below also hold for the standard (truncated) version of heatmaps.

\paragraph{Bound on EMD.} We start with the simplest bound: that of EMD on the (infinite) heatmaps.

\begin{lemma} \label{lem:heatmapEMD}
For any $\bp, \bq \in \R^{G_{\Delta}}_{\geq 0}$ such that $\|\bp\|_1 = \|\bq\|_1$ and $\sigma > 0$, we have $$\emd(\tH^\sigma_\bp, \tH^\sigma_\bq) \leq \emd(\bp, \bq).$$
\end{lemma}

\begin{proof}
Let $\gamma$ be such that
\begin{align*}
\emd(\bp, \bq) = \sum_{a \in G_\Delta} \sum_{b \in G_\Delta} \gamma(a, b) \cdot \|a - b\|_1,
\end{align*}
and $\gamma$'s marginals are equal to $\bp, \bq$. For convenience, we let $\gamma(a, b)$ be zero when $a \notin G_\Delta$ or $b \notin G_\Delta$.
We then define $\tgamma$ on $\tG_\Delta \times \tG_\Delta$ by
\begin{align*}
\tgamma(a, b) = \sum_{d \in \tG_{\Delta}} \frac{1}{Z} \cdot e^{-\frac{\|d\|_2^2}{2\sigma^2}} \cdot \gamma(a + d, b + d).
\end{align*}
It is simple to check that the marginals of $\tgamma$ are exactly $\tH^\sigma_\bp, \tH^\sigma_\bq$. As such, we have
\begin{align*}
& \emd(\tH^\sigma_\bp, \tH^\sigma_\bq) \\
&\leq \sum_{a \in \tG_\Delta} \sum_{b \in \tG_\Delta} \tgamma(a, b) \cdot \|a - b\|_1 \\
&= \sum_{a \in \tG_\Delta} \sum_{b \in \tG_\Delta} \|a - b\|_1 \cdot \left(\sum_{d \in \tG_{\Delta}} \frac{1}{Z} \cdot e^{-\frac{\|d\|_2^2}{2\sigma^2}} \cdot \gamma(a + d, b + d)\right) \\
&= \sum_{a', b' \in \tG_\Delta} \gamma(a', b') \cdot \|a' - b'\|_1 \cdot \left(\sum_{d \in \tG_\Delta} \cdot \frac{1}{Z} \cdot e^{-\frac{\|d\|_2^2}{2\sigma^2}}\right) \\
&= \emd(\bp, \bq).
\end{align*}
where the penultimate equality is from substituting $a' = a + d$ and $b' = b + d$.
\end{proof}

\paragraph{Bound on KL-divergence.}
We will now move on to the KL-divergence. For this purpose, we will use the following bound, shown in~\citet{Canonne0S20}\footnote{Strictly speaking,~\citet{Canonne0S20} only states the lemma for the 1-dimensional case but it is not hard to see that it holds for the 2-dimensional case of our interest as well.}. Here we use $\bone_p$ to denote the indicator variable for the point $p$, i.e., a vector that is equal to one at coordinate $p$ and zero elsewhere.

\begin{lemma}[{\cite[Proposition 5]{Canonne0S20}}] \label{lem:kld-discrete-gaussian}
For any $a, b \in G_\Delta$ and $\sigma > 0$, we have
\begin{align*}
\KL(\tH_{\bone_a}^\sigma \| \tH_{\bone_b}^\sigma) \leq \frac{\|a - b\|_2^2}{2\sigma^2}.
\end{align*}
\end{lemma}

We can now prove that EMD on the distribution aggregation error implies a KL-divergence bound of the (infinite) heatmaps:

\begin{lemma}
\label{lem:heatmapfromEMD}
For any $\bp, \bq \in \R^{G_{\Delta}}_{\geq 0}$ such that $\|\bp\|_1 = \|\bq\|_1 = 1$ and $\sigma > 0$, we have $$\KL(\tH^\sigma_\bp \| \tH^\sigma_\bq) \leq \frac{\emd(\bp, \bq)}{2\sigma^2}.$$
\end{lemma}

\begin{proof}
Let $\gamma$ be such that
\begin{align*}
\emd(\bp, \bq) = \sum_{a \in G_\Delta} \sum_{b \in G_\Delta} \gamma(a, b) \cdot \|a - b\|_1,
\end{align*}
and $\gamma$'s marginals are equal to $\bp, \bq$.

From the convexity of KL-divergence, we have
\begin{align*}
\KL(\tH^\sigma_\bp, \tH^\sigma_\bq) 
&\leq \sum_{a \in G_\Delta} \sum_{b \in G_\Delta} \gamma(a, b) \cdot \KL(\tH^\sigma_{\bone_a}, \tH^\sigma_{\bone_b}) \\
(\text{\Cref{lem:kld-discrete-gaussian}}) &\leq \sum_{a \in G_\Delta} \sum_{b \in G_\Delta} \gamma(a, b) \cdot \frac{\|a - b\|_2^2}{2\sigma^2} \\
(\text{From } a, b \in [0, 1)) &\leq \sum_{a \in G_\Delta} \sum_{b \in G_\Delta} \gamma(a, b) \cdot \frac{\|a - b\|_1}{2\sigma^2} \\
&\leq \frac{\emd(\bp, \bq)}{2\sigma^2}. \qedhere
\end{align*}
\end{proof}

\paragraph{Bound on TV distance and Similarity.}
Via Pinsker's inequality and our previous bound on the KL-divergence, we arrive at the following:
\begin{lemma}
\label{lem:heatmapTVfromEMD}
For any $\bp, \bq \in \R^{G_{\Delta}}_{\geq 0}$ such that $\|\bp\|_1 = \|\bq\|_1 = 1$ and $\sigma > 0$, we have $$\TV(\tH^\sigma_\bp \| \tH^\sigma_\bq) \leq \frac{\sqrt{\emd(\bp, \bq)}}{2\sigma}.$$
\end{lemma}

Finally, recall that the similarity (SIM) metric is simply one minus the total variation distance; this, together with the above lemma, leads to the following bound:
\begin{corollary}
For any $\bp, \bq \in \R^{G_{\Delta}}_{\geq 0}$ such that $\|\bp\|_1 = \|\bq\|_1 = 1$ and $\sigma > 0$, we have $$\SIM(\tH^\sigma_\bp \| \tH^\sigma_\bq) \geq 1 - \frac{\sqrt{\emd(\bp, \bq)}}{2\sigma}.$$
\end{corollary}
\section{A Tight Lower Bound for Sparse EMD Aggregation}
\label{app:sparselb}

Our lower bounds in this and the subsequent section apply for any sufficiently large $\Delta$ (depending on the other parameters); this will be assumed without explicitly stated in the subsequent statement of the theorems/lemmas. We remark that this is necessary because if $\Delta$ is too small, then one can get smaller error. (E.g., if $\Delta = 1$, it is obvious how to get $O(1/\eps)$ error.)

The main result of this section is a lower bound of $\Omega_\eps(\sqrt{k} / n)$ on the additive error of any $\eps$-DP algorithm for sparse EMD aggregation, which matches the error bounds we achieved in our algorithm (\Cref{thm:main-ub}):

\begin{theorem} \label{thm:packing-lb}
For any $\lambda, \eps > 0$, any integers $k \geq \eps$ and $n \geq k / \eps$, no $\eps$-DP algorithm can, on input consisting of $n$ distributions, output a $\left(\lambda, o\left(\frac{\sqrt{k}}{\eps n}\right)\right)$-approximation for sparse EMD aggregation, with probability at least 0.1.
\end{theorem}

We prove the above result via the packing framework of Hardt and Talwar~\cite{HardtT10}.
Recall that we say that a set of $k$-sparse probability distributions $\bp_1, \dots, \bp_T$ is a \emph{$\gamma$-packing} (for EMD distance) if $\emd(\bp_i, \bp_j) \geq \gamma$ for all $i, j \in \{1, \dots, T\}$; the \emph{size} of the packing is $T$. The following lemma provides such a packing with large size, based on constructing an $\ell_1$ packing on a $\sqrt{k} \times \sqrt{k}$ grid.

\begin{lemma} \label{lem:emd-packing}
For any $\gamma > 0$ such that $\gamma < 0.01 / \sqrt{k}$, there exists a set of $k$-sparse probability distributions that forms a $\gamma$-packing of size $2^{\Omega(k \cdot \log(1 / (\gamma \sqrt{k}))}$. Furthermore, we can pick the packing so that each of its elements (i.e., distributions) has the same support.
\end{lemma}

\begin{proof}
Consider the grid points $G_t$ where $t$ is the smallest power of two that is no larger than $\sqrt{k}$. Let $\cH$ denote any maximal subset of $S := \{\br \in \R^{G^t}_{\geq 0} \mid \|\br\|_1 = 0.5\}$ such that the pairwise $\ell_1$ distance of elements of $\cH$ is at least $2\gamma t$. (In other words, $\cH$ is a $(2\gamma t)$-packing of $S$ under $\ell_1$ distance.) Due to the maximality of $\cH$, the $\ell_1$ balls of radius $2\gamma t$ around elements of $\cH$ must cover $S$. Recall that the volume of $d$-dimensional $\ell_1$ ball of radius $r$ is $(2r)^d / d!$. Thus, we have
\begin{align*}
|\cH| \geq \frac{\vol(S)}{(4\gamma t)^d / d!} = \frac{0.5^d / d!}{(4\gamma t)^d / d!} \geq 2^{\Omega(k \cdot \log(1/(\gamma \sqrt{k})))},
\end{align*}
where in the equality we use the fact that the $d$-simplex has volume $1/d!$ and in the second inequality we use the fact that $t = \Theta(\sqrt{k})$ and $d = t^2 = \Theta(k)$.

Finally, we create our packing $\cP$ as follows: for every $\br \in \cH$, add to $\cP$ the distribution $\bp$ defined by $\bp(a) = \br(a) + 0.5 / t^2$ for all $a \in G_t$.
Now, notice that, for any pair of distributions $\bp, \bq$ over $G_t$, we have $\emd(\bp, \bq) \geq \|\bp - \bq\| / (2t)$ because any two points in $G_t$ are of distance at least $1/t$ apart; thus $\cP$ is indeed a $\gamma$-packing. Furthermore, since $|G_t| \leq t^2 \leq k$, any probability distribution over $G_t$ is $k$-sparse as desired.
\end{proof}

We can now apply the packing framework~\cite{HardtT10} to prove \Cref{thm:packing-lb}:

\begin{proof}[Proof of \Cref{thm:packing-lb}]
Let $\gamma = \Omega(1 / \sqrt{k})$ denote the largest $\gamma$ for which \Cref{lem:emd-packing}  guarantees the existence of $\gamma$-packing $\cP = \{\bp_1, \dots, \bp_T\}$ of size (at least) $10 \cdot e^k$. Let $n_0 := \lfloor k/\eps \rfloor$, and consider any $n \geq n_0$.

Suppose for the sake of contradiction that there is an $\eps$-DP algorithm $\A$ that with probability at least 0.1 outputs an $(\lambda, 0.4\gamma n_0 / n)$-approximation. (Note that, since $k \geq \eps$, $0.4\gamma n_0 / n = \Omega\left(\frac{\sqrt{k}}{\eps n}\right)$.) Let $\bp^*$ be any element of $\cP$. For every distinct $\bp \in \cP$, let $B_\bp$ denote the set of all $\bq$ such that $\emd\left(\bq, \frac{n_0}{n} \cdot \bp + \frac{n - n_0}{n} \cdot \bp^*\right) \leq 0.4\gamma$. Furthermore, let $\bX_{\bp}$ denote the dataset consisting of $n_0$ copies of $\bp$ and $n - n_0$ copies of $\bp^*$.
Note that the average distribution of $\bX_\bp$ is $k$-sparse, since every distribution in $\cP$ is $k$-sparse and has the same support. Hence, from the assumed accuracy guarantee of $\A$, we have
\begin{align*}
\Pr[\A(\bX_\bp) \in B_\bp] \geq 0.1.
\end{align*}
Next, let $\bX^*$ denote the dataset consisting of $n - n_0$ copies of $\bp^*$. From the $\eps$-DP guarantee of $\A$, we have
\begin{align} \label{eq:packing-each-input}
\Pr[\A(\bX^*) \in B_\bp] \geq e^{-\eps n_0} \cdot \Pr[\A(\bX_\bp)] \geq 0.1 \cdot e^{-k}.
\end{align}
Now, notice that $B_\bp$ are disjoint for all $\bp \in \cP$ because $\cP$ forms a $\gamma$-packing. Thus, we have
\begin{align*}
\Pr\left[\A(\bX^*) \in \bigcup_{\bp \in \cP} B_\bp\right] & = \sum_{\bp \in \cP} \Pr\left[\A(\bX^*) \in B_\bp\right] \\
& \overset{\eqref{eq:packing-each-input}}{\geq} |\cP| \cdot 0.1 \cdot e^{-k} > 1,
\end{align*}
a contradiction.
\end{proof}

\section{Dense EMD Aggregation: Tight Bounds}

In this section, we consider the distribution aggregation under EMD but without any sparsity constraint. 
%
We will show that in this case the tight EMD error is $\tilde{\Theta}(1/\sqrt{\eps n})$. Notice that this is a factor $\tilde{\Theta}_\eps(\sqrt{n/k})$ larger than the additive error in our sparse EMD aggregation (\Cref{thm:main-ub}).

\subsection{Algorithm}

In this section, we present a simple algorithm for EMD aggregation that yields the claimed bound.

\begin{theorem}
For any $\eps > 0$, there exists an $\eps$-DP algorithm for EMD aggregation that with probability 0.99 incurs an expected error of $\tilde{O}(1/\sqrt{\eps n})$.
\end{theorem}

\begin{proof}
Let $\ell^* = \log \sqrt{\eps n}$ and let $\Delta^* = 2^{\ell^*}$. Our algorithm starts by snapping every point to the closest point in $t$; more formally, we consider user $i$'s input to be $P_t \bp_i$ instead of the original $\bp_i$. This can contribute to at most $\sqrt{\eps n}$ in the additive error, and henceforth we may assume that each input $\bp_i$ belongs to $\R_{\geq 0}^{G_{\Delta^*}}$ (i.e., $\Delta = \Delta^*$) instead.

From here, we just compute the sum $\bs := \sum_{i=1}^n \bp_i$ and add Laplace noise to get $\tbs = \bs + \bnu$ where $\bnu(a) \sim \Lap(1/\eps)$ for each $a \in G_{\Delta}$. Then, we find $\hbs \in \R_{\geq 0}^{G_t}$ that minimizes $\|\hbs - \tbs\|_{\emd}$. (This can be formulated as a linear program.) Finally, output $\hba = \hbs / \|\hbs\|_1$.

It is straightforward that the algorithm is $\eps$-DP, since it consists of only applying the Laplace mechanism once and post-processing its result. Next, we will analyze its utility. From \Cref{lem:err-from-normalization} and Markov's inequality, it suffices to show that $\E[\|\hbs - \bs\|_{\emd}] \leq \tilde{O}(\sqrt{n/\eps})$. Due to the triangle inequality, we have
\begin{align*}
\|\hbs - \bs\|_{\emd} & \leq \|\hbs - \tbs\|_{\emd} + \|\tbs - \bs\|_{\emd} \\
& \leq 2 \cdot \|\tbs - \bs\|_{\emd} = 2\|\bnu\|_{\emd},
\end{align*}
where the latter inequality follows from how $\hbs$ is computed. As a result, it suffices to show that $\E[\|\bnu\|_{\emd}] \leq \tilde{O}(\sqrt{n / \eps})$, which we do next.

From \Cref{lem:emd-to-l1-expansion}, we have
\begin{align*}
& \E[\|\bnu\|_{\emd}] \\
& \leq \E[\|\bP\bnu\|_1] \\
&= \sum_{i \in [\ell^*]} 2^{-i} \E[\|\bP_i\bnu\|_1] \\
&= \sum_{i \in [\ell^*]} 2^{-i} \sum_{c \in C^{2^i}} \E\left[\left|\sum_{p \in c \cap G_{\Delta}} \bnu(p) \right|\right] \\
\intertext{by the Cauchy--Schwarz inequality}
&\leq \sum_{i \in [\ell^*]} 2^{-i} \sum_{c \in C^{2^i}} \sqrt{\E\left[\left(\sum_{p \in c \cap G_{\Delta}} \bnu(p) \right)^2\right]} \\
\intertext{by the independence of $\bnu(i)$}
&= \sum_{i \in [\ell^*]} 2^{-i} \sum_{c \in C^{2^i}} \sqrt{\E\left[\sum_{p \in c \cap G_{\Delta}} \bnu(p)^2\right]} \\
&= \sum_{i \in [\ell^*]} 2^{-i} \sum_{c \in C^{2^i}} \sqrt{|p \in c \cap G_{\Delta}| \cdot \frac{2}{\eps^2}} \\
&\leq O\left(\sum_{i \in [\ell^*]} 2^{-i} \sum_{c \in C^{2^i}} 2^{\ell^* - i} \cdot \frac{1}{\eps} \right) \\
&= O\left(\sum_{i \in [\ell^*]} 2^{-i} \cdot 2^{2i} \cdot 2^{\ell^* - i} \cdot \frac{1}{\eps} \right) \\
&= O\left(\ell^* \cdot 2^{\ell^*} / \eps\right) \\
&= O\left(\log(\eps n) \cdot \sqrt{\eps n} / \eps\right) \\
&= \tilde{O}(\sqrt{n / \eps}). \qedhere
\end{align*}
\end{proof}

\subsection{Lower Bound}

A matching lower bound can be immediately obtain from plugging in $k = \lfloor \eps n \rfloor$ into our lower bound for the sparse case (\Cref{thm:packing-lb}). This immediately yields the following.

\begin{corollary}
For any $\eps > 0$ and any integer $n \geq \eps$, no $\eps$-DP algorithm can incur an error at most $o(1/\sqrt{\eps n})$ for EMD aggregation with probability at least 0.1.
\end{corollary}

\section{Additional Experiment Results}
\label{app:exp}

\subsection{Additional Results on Geo Datasets}

Recall the second (varying the number of users) and third (varying the resolution) experimental setups described in Section~\ref{subsec:exp-results-main}. We had presented the results on the EMD metric in \Cref{fig:user-resolution}. Results on other metrics can be found in \Cref{fig:varying-user} and \Cref{fig:varying-resolution}, respectively. Other parameters are fixed as before, e.g., $\eps = 10$. The trends are similar across all metrics. 

As predicted by theory, the utility of both the baseline and our algorithm increases with $n$, the number of users.  On the other hand, the threshold variants of the baseline initially performs better as $n$ increases but after some point it starts to become worse. A possible reason is that, as $n$ increases, the sparsity level (i.e., the number of non-zero cells after aggregation) also increases, and thus thresholding eventually becomes too aggressive and loses too much information.

\begin{figure*}
\centering
\includegraphics[trim={6cm 0 6cm 0},clip,scale=0.275]{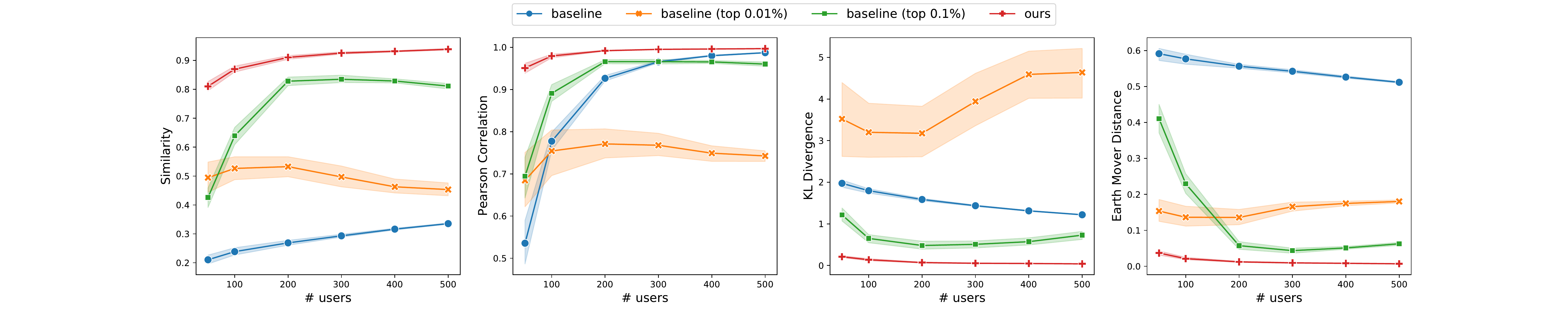}
\caption{Effect of the \# users on different metrics ($\eps = 10$). Shaded areas indicate 95\% confidence interval. 
\label{fig:varying-user}}
\end{figure*}

In terms of the resolution $\Delta$, our algorithm's utility remains nearly constant across all metrics as $\Delta$ increases, whereas the baseline's utility suffers. The performance of the thresholding variants of the baseline is less predictable. Both the top-0.01\% and top-0.001\% variants start off with lower utility compared to baseline, which might be due to the fact, when $\Delta$ is small, thresholding leaves too few grid cells with non-negative values. On the other hand, the top-0.01\% variant's utility also starts to suffer as $\Delta$ increases from 256 to 512. This may represent the opposite problem: top-0.01\% leaves too many grid cells with noisy positive values, resulting in the decrease in utility.

\begin{figure*}
\centering
\includegraphics[trim={6cm 0 6cm 0},clip,scale=0.275]{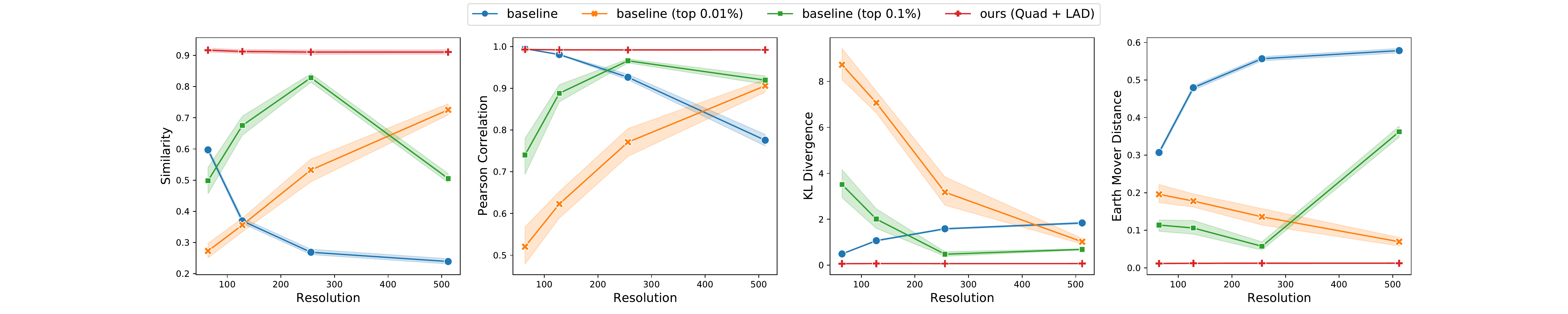}
\caption{Effect of the resolution on different metrics ($\eps = 10$). Shaded areas indicate 95\% confidence interval. 
\label{fig:varying-resolution}}
\end{figure*}

\begin{figure*}
\centering
\includegraphics[trim={6cm 0 6cm 0},clip,scale=0.275]{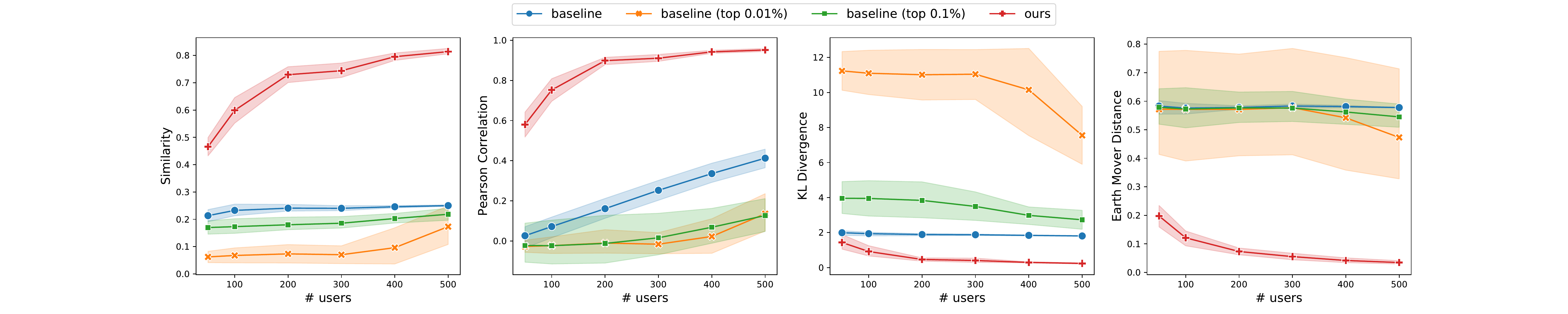}
\caption{Effect of the \# users on different metrics ($\eps = 1$). Shaded areas indicate 95\% confidence interval. 
\label{fig:varying-user-eps1}}
\end{figure*}

\begin{figure*}
\centering
\includegraphics[trim={6cm 0 6cm 0},clip,scale=0.275]{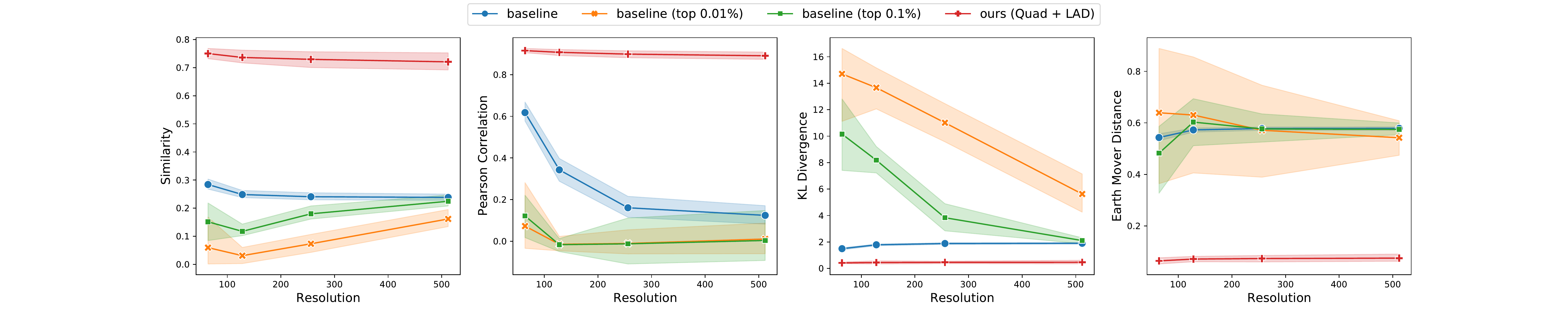}
\caption{Effect of the resolution on different metrics ($\eps = 1$). Shaded areas indicate 95\% confidence interval. 
\label{fig:varying-resolution-eps1}}
\end{figure*}

We also include the results for similar setups but with privacy parameter $\eps = 1$ in \Cref{fig:varying-user-eps1} and \Cref{fig:varying-resolution-eps1}. The trends of the vanilla baseline and our algorithm are identical when $\eps = 10$. The thresholding variants however have low utility throughout. In fact, for most metrics and parameters, they are worse than the original baseline; it seems plausible that, since the noise is much larger in this case of $\eps = 1$, top grid cells kept are mostly noisy and thus the utility is fairly poor.

\subsection{Results on the Salicon dataset}
\label{subsec:salicon-exp}

We also experiment on the Salicon image saliency dataset~\cite{jiang2015salicon}, \salicon, available from \url{http://salicon.net/}.  This dataset is a collection of saliency annotations on the popular Microsoft Common Objects in Context (MS COCO) image database, where each image has $640 \times 480$ pixels.  For the purposes of our experiment, we assume the data is of the following form: for each (user, image) pair, it consists of a sequence of coordinates in the image where the user looked at.  

We repeat the first set of experiments on 38 randomly sampled images (with $\geq 50$ users each) from \salicon. The natural images were downsized to a fixed resolution of $\Delta_1 \times \Delta_2 = 320 \times 240$. For each $\eps \in \{1, 2, 5, 8, 10, 20\}$, we run our algorithms together with the baseline and its variants on all 38 natural images, with 2 trials for each image. In each trial, we sample a set of 50 users and run all the algorithms; we then compute the metrics between the true heatmap and the estimated heatmap. 

The average of these metrics over the 76 runs is presented in~\Cref{fig:supp-comparison-result}, together with the 95\% confidence interval. As can be seen in the figure, the baseline has rather poor performance across all metrics, even for large $\eps = 20$. We experiment with several values of $t$ for the thresholding variant and still observe an advantage of our algorithm consistently across all metrics. We also include examples of the resulting heatmaps with original image overlaid from each approach in ~\Cref{fig:supp-visualization}.

\begin{figure*}[h]
\centering
\includegraphics[trim={6cm 0 6cm 0},clip,scale=0.275]{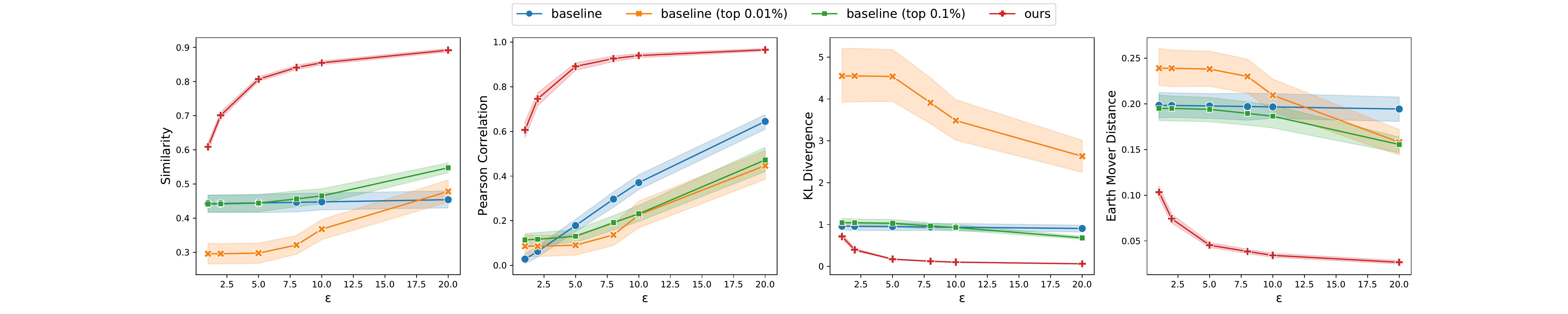}
\caption{Metrics averaged over 76 runs when varying $\eps$. Shaded areas indicate 95\% confidence interval.
\label{fig:supp-comparison-result}}
\end{figure*}

\begin{figure*}[h]
\centering
\includegraphics[trim={7cm 1.5cm 6cm 1.5cm},clip,scale=0.275]{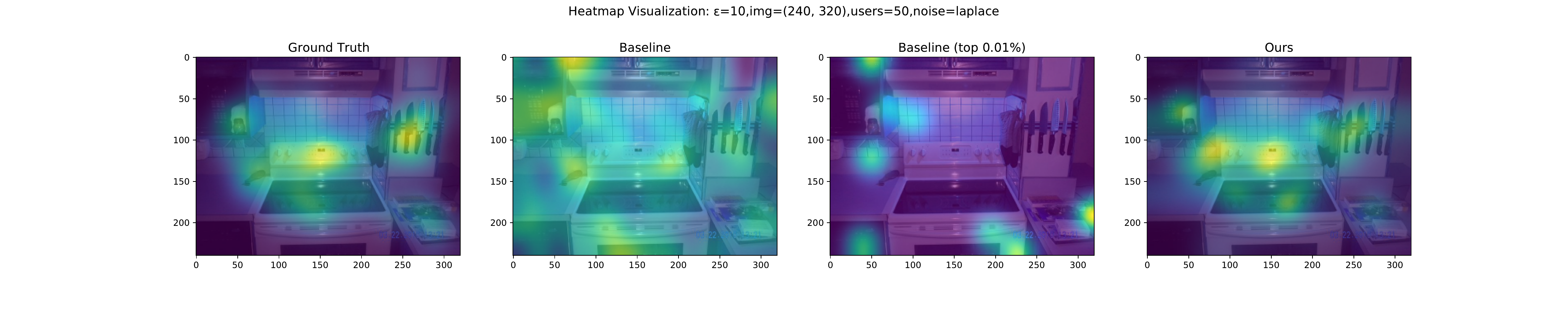} 
\includegraphics[trim={7cm 1.5cm 6cm 1.5cm},clip,scale=0.275]{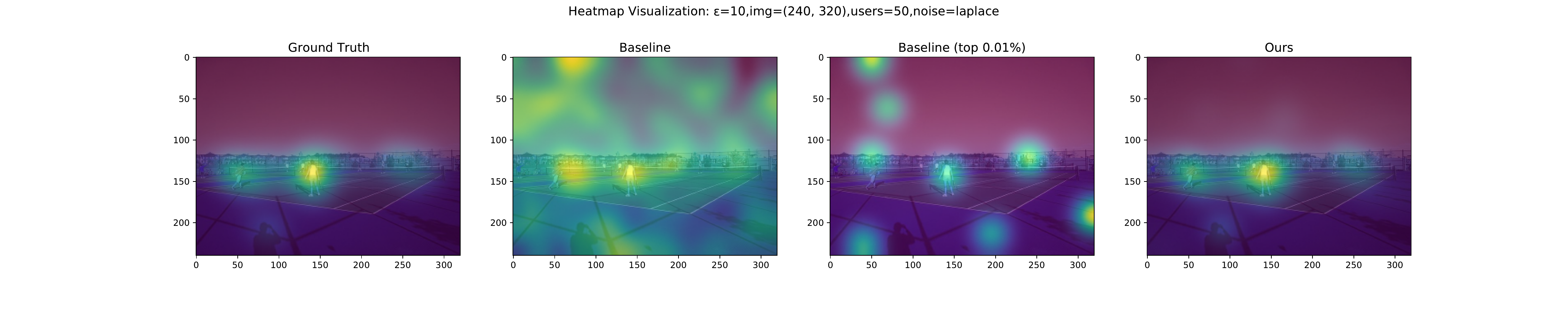} 
\caption{Example visualization of different algorithms for two different natural images from the Salicon dataset (top and bottom rows) for $\eps=10$ and $n=50$ users. The algorithms from left to right are: original heatmap (no privacy), baseline, baseline with top 0.01\% and our algorithm. \label{fig:supp-visualization}}
\end{figure*}

\subsection{Results on Synthetic Data}

Finally, we experiment on synthetically generated datasets. We generate the  data by sampling from a uniform mixture of two-dimensional Gaussian distributions, while the mean and covariance matrix of each Gaussian is randomly generated. The number of Gaussians, their covariances, and the  number of samples from the mixture, are changed to control the sparsity of data.  
Throughout the experiments, we fix the number of users to 200, $\eps = 1$, and use $w = 20$.

We run two sets of experiments to demonstrate the effect of sparsity levels. In the first set of experiments, we fix the number of clusters to $20$. We then sample points from this Gaussian mixture until we reach different sparsity values. For each sparsity level, we generate 10 datasets. The experiment results are shown in \Cref{fig:sparsity-v-utility-20-clusters}. (Note that we bucketize the sparsity level.)

In the second set of experiments, we vary both the sparsity levels and the number of Gaussians in the mixture. Specifically, we use $5$ or $10$ Gaussians to produce low sparsity level $< 5\%$, $20$ clusters to produce mid-level sparsity like $\unsim 15\%$ and $\unsim 30\%$, and $80$ clusters to produce sparsity level like $\unsim 75\%$. 
The results for this set of experiments are shown in \Cref{fig:sparsity-v-utility-varying-clusters}.

In both experiments, the rough trends are fairly similar: our algorithm's utility decreases as the sparsity level increases, wheres the baseline's utility increases. Upon closer inspection, however, the rate at which our algorithm's utility decreases is slower when we fixed the number of Gaussians. Indeed, this should be unsurprising: our guarantee is with respect to the best $k$-sparse approximation (in EMD), which remains roughly constant for the same mixture of Gaussians, regardless of the sparsity. As such, the utility does not decrease too dramatically if we keep the mixture distribution the same. On the other hand, when we also increase the number of Gaussians in the mixture, even the $k$-sparse approximation fails to capture it well\footnote{Note that while we do not choose $k$ explicitly, it is governed by the parameter $w$, which we fix to $20$ throughout.}, resulting in faster degradation in utility.

\begin{figure*}
\centering
\includegraphics[trim={0cm 0 0cm 0},clip,scale=0.275]{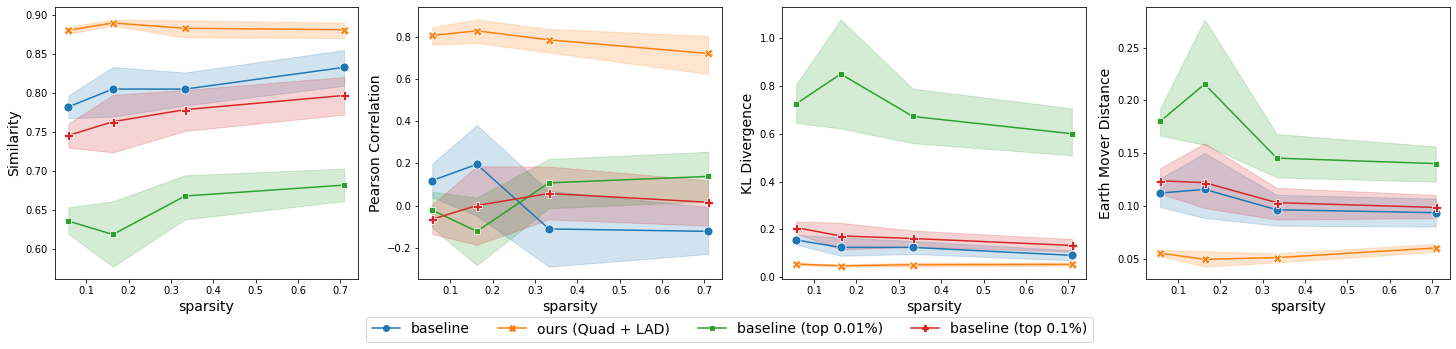}
\caption{Metrics averaged over 10 runs when varying the sparsity of the dataset; the number of Gaussians in the mixtures is kept constant at $20$ respectively. Shaded areas indicate 95\% confidence interval.
\label{fig:sparsity-v-utility-20-clusters}}
\end{figure*}

\begin{figure*}
\centering
\includegraphics[trim={0cm 0 0cm 0},clip,scale=0.275]{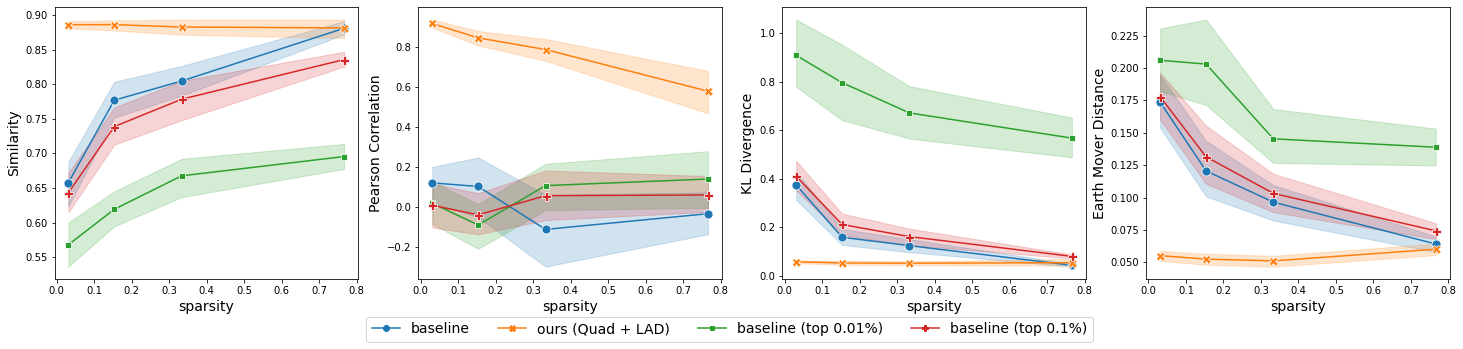}
\caption{Metrics averaged over 10 runs when varying both the sparsity of the dataset and the number of Gaussians in the mixture; the number of Gaussians is increased from 5 to 10, 20, and 80. Shaded areas indicate 95\% confidence interval.
\label{fig:sparsity-v-utility-varying-clusters}}
\end{figure*}

\section{Implementation in Shuffle Model}

As alluded to earlier, our algorithm can be easily adapted to any model of DP that supports a ``Laplace-like'' mechanism. In this section, we illustrate this by giving more details on how to do this in the shuffle model of DP~\cite{bittau17,erlingsson2019amplification,CheuSUZZ19}. Recall that, in the shuffle model, each user produces messages based on their own input. These messages are then randomly permuted by the shuffler and forwarded to the analyzer, who then computes the desired estimate. The DP guarantee is enforced on the shuffled messages.

To implement our algorithms in the shuffle model, we may use the algorithms of~\citet{balle_merged,ghazi2019private} to compute $\by'$ in \Cref{alg:dp-aggregation-full} (instead of running the central DP algorithm); the rest of the steps remain the same. To compute the estimate $\by'$, the algorithm from~\citet{balle_merged} works as follows where $B \in \N$ is a parameter (recall that $\Z_q := \{0, \dots, q - 1\}$ and the summation is modulo $q$):
\begin{itemize}
\item Let $q = Bn, m = m_0 + \cdots + m_\ell$ and\footnote{Note that the parameter is set as in Corollary 6 and Lemma 5.2 of \citet{balle_merged} to achieve $(\eps, \delta)$-shuffle DP.} $$r = \left\lceil \frac{2\ln(e^\eps + 1) + 2\ln(m/\delta) + \ln q}{\ln n} + 1\right\rceil.$$
\item Each user $j$: 
\begin{itemize}
\item Compute $\by^j := [\bP_0 \bp_j \bP_2 \bp_j \cdots \bP_\ell \bp_j]$.
\item Scaled \& round: $\bz^j := \lfloor B\by^j \rfloor$ where the floor is applied coordinate-wise.
\item Add noise: $$\bz'^j := \bz^j + \bigotimes_{i \in [\ell]} \Polya(1/n, e^{-\eps_i})^{\otimes m_i}.$$
\item Split \& Mix: For  sample $\bw^1, \dots, \bw^{r - 1}$ i.i.d. at random from $\Z_q^m$ and then let $\bw^r = \bz^j - \bw^1 - \cdots - \bw^{r - 1}$.
\item For each $a \in [m]$ and $b \in [r]$, send a tuple $(a, w^b_a)$ to the shuffler\footnote{Note that this ``coordinate-by-coordinate'' technique is from~\cite{CheuSUZZ19}.}.
\end{itemize}
\item Once the analyzer receives all the tuples, for each $a \in [m]$, they sum second components of all messages whose first component is $a$ up (over $\Z_q$), and let the $a$th coordinate of $\bz'$ be equal to the sum. Finally, let $\by' = \bz' / B$.
\end{itemize}
We note that the Polya distribution and its parameters are set in such a way that the sum of all the noise is exactly the discrete Laplace distribution. Please refer to~\cite{balle_merged} for more detail.

Choosing the parameter $B$ will result in a tradeoff between the utility and the communication complexity. For smaller $B$, the communication will be smaller but there will be more error from rounding. On the other hand, larger $B$ results in larger communication but more accurate simulation of the central DP mechanism. (In fact, as $B \to \infty$, the resulting estimate $\by'$ converges in distribution to that of the central DP mechanism.)

\subsection*{Empirical Results}
We demonstrate the effect of choosing $B$ on communication and utility by running an experiment on the \salicon dataset (where the setup is as in~\Cref{subsec:salicon-exp}) with fixed $n = 50$ and $\eps=5, \delta = 10^{-5}$.

\begin{figure}[h!]
\centering
\includegraphics[width=0.3\textwidth]{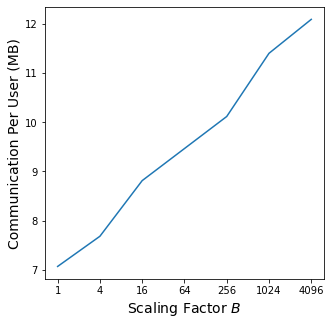}
\caption{Communication per user in shuffle DP protocol when varying the scaling factor $B$. 
\label{fig:comm-shuffle}}
\end{figure}

\begin{figure*}[h!]
\centering
\includegraphics[trim={3cm 0 0 0},clip,scale=0.3]{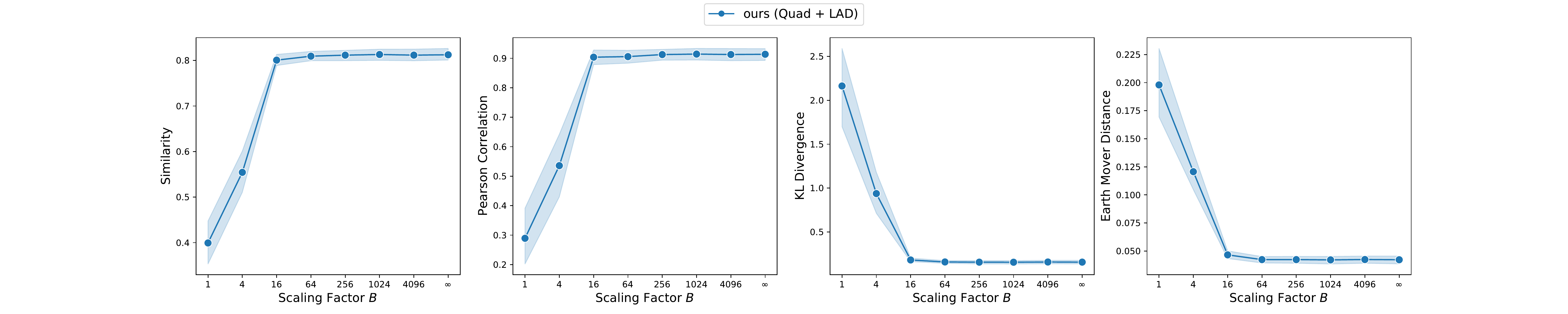}
\caption{Metrics averaged over 20 runs when varying the scaling factor $B$. Shaded areas indicate 95\% confidence interval. We use the convention $B = \infty$ for the central DP algorithm (which does not apply any rounding). 
\label{fig:util-shuffle}}
\end{figure*}

\subsubsection*{Communication} We start with the communication. Recall from the aforementioned algorithm that each user sends $r m$ messages, each of length $\lceil \log_2(mq) \rceil$ bits. This gives the total communication per user of $rm \lceil \log_2(mq) \rceil$ bits. The concrete numbers are shown in \Cref{fig:comm-shuffle}. Note that this is independent of the input datasets, except for the resolution (which effects $m$). Recall from earlier that we use resolution of $320 \times 240$ for the \salicon dataset.

\subsubsection*{Utility} As for the utility, we run the experiments on the \salicon dataset where the setup is as in~\Cref{subsec:salicon-exp}. The utility metrics are reported in \Cref{fig:util-shuffle}. In all metrics, the utility is nearly maximized already when $B$ is no more than 256. As per \Cref{fig:comm-shuffle}, this would correspond to approximately 10MB communication required per user.
\fi

\end{document}